\documentclass[11pt]{article}
\usepackage{amssymb, amsthm, amsmath}
\usepackage{booktabs}
\usepackage{multirow}
\usepackage{graphicx}
\usepackage{natbib}
\usepackage{subfigure}
\usepackage{url}
\usepackage[left=25.4mm,right=25.4mm,top=30mm,bottom=25.4mm,a4paper]{geometry}
\theoremstyle{definition}
\newtheorem{theorem}{Theorem}
\newtheorem{lemma}[theorem]{Lemma}

\newtheorem{remark}[theorem]{Remark}

\newtheorem{example}{Example}

\newcommand{\E}{\mathbb E}
\newcommand{\Q}{\mathbb Q}
\newcommand{\D}{\mathrm{d}}

\newcommand{\F}{\mathcal F}
\newcommand{\e}{\mathrm{e}}
\renewcommand{\cite}{\citet}

\begin{document}
\title{High moment variations and their application}
\author{Geon Ho Choe\footnote{
Department of Mathematical Sciences, KAIST, Daejeon 305-701, Korea, choe@euclid.kaist.ac.kr}
and Kyungsub Lee\footnote{Correspondence author, Department of Mathematical Sciences, KAIST, Daejeon 305-701, Korea, Tel: +82-42-350-2725, Fax: +82-42-350-2710, klee@euclid.kaist.ac.kr}
}

\date{}
\maketitle

\begin{abstract}
We propose a new method of measuring the third and fourth moments of return distribution based on quadratic variation method when the return process is assumed to have zero drift.
The realized third and fourth moments variations computed from high frequency return series are good approximations to corresponding actual moments of the return distribution.
An investor holding an asset with skewed or fat-tailed distribution is able to hedge the tail risk by contracting the third or fourth moment swap under which the float leg of realized variation and the predetermined fixed leg are exchanged.
Thus constructed portfolio follows more Gaussian-like distribution and hence the investor effectively hedge the tail risk.
\end{abstract}

\section*{Acknowledgement}
This research was supported by Basic Science Research Program through the National Research Foundation of Korea(NRF) funded by the Ministry of Education(NRF-2011-0012073).

\newpage

\section{Introduction}

We define the third and fourth moment variations of financial asset return process and examine the properties of the variations.
The third moment variation is defined to be a quadratic covariation between squared return and return process and
the fourth moment variation is defined to be a quadratic variation of squared return process.
It is demonstrated that the variations can be used as approximations for the third and fourth moments of the return distribution, which are generally hard to measure, under certain conditions.

Skewness, the third standardized moment, has long been an important topic in financial study and
there are numerous evidences that stock return distributions are left skewed under both physical and risk-neutral probability.
\cite{KrausLitzenberger} extended the capital asset pricing model to incorporate skewness preference.
\cite{HarveySiddique1999} developed a new methodology for estimating time-varying conditional skewness.
\cite{HarveySiddique} showed that conditional skewness explain the cross-sectional variation of expected returns across assets.
In \cite{BakshiKapadiaMadan}, cubic and quartic contracts are defined to measure risk-neutral skewness and kurtosis.
\cite{Christoffersen} developed a GARCH type option pricing model with inverse Gaussian innovations to incorporate conditional skewness.
\cite{Neuberger2012} and \cite{Kozhan} proposes a new definition of realized third moment which satisfies the aggregate property to estimate of the true third moment of long-horizon returns.

In addition, kurtosis, the fourth standardized moment, also plays an crucial role in financial studies
and it is well known that a financial asset return distribution has larger kurtosis than the Gaussian distribution.
\cite{Brooks} proposed a new model for autoregressive kurtosis and showed the evidence of the presence of autoregressive conditional kurtosis.
\cite{LeonRubioSerna} also indicated a significant presence of conditional skewness and kurtosis.

Our approaches to examining the high moments of return distribution are based on quadratic variation methods
and the motivation stems from the previous researches on the quadratic variation of asset return process.
The quadratic variation of an asset return (or price) process, the limit of the sum of squared returns, plays a central role in measuring the variance of return, since the expectation of the quadratic variation of return is considered as an estimator of the variance of the return distribution.
The realized (quadratic) variation, which refers to the finite sum of squared return computed from high-frequency return time series, is an approximation to the quadratic variation.
Thus, the realized variation of return is an efficient estimator of the variance of return distribution under physical probability.
We extend this idea to link newly defined high moment variations and the corresponding quantities of the return distribution.

In the previous studies about theory and empirical analysis on high-frequency data,
\cite{ABDL} showed that when underlying process is a semimartingale, the realized variance is a consistent estimator of quadratic variation.
\cite{Barndorff2002a} derived the asymptotic distribution of the realized volatility error under stochastic volatility models.
\cite{Barndorff-Nielsen2004} and \cite{Barndorff-Nielsen2006} introduced realized bipower variation which is robust to rare jumps in estimating integrated variance and tested jump in asset price process.
\cite{Hansen} examined the errors of the realized variance under the presence of market microstructure noise.
\cite{MyklandZhang} studied local-constancy approximation of variance on which econometric literatures of high frequency data often rely.

One of the interesting properties of the quadratic variation of the return is that the risk-neutral expectation of the quadratic variation is synthesized by a continuum of European option prices.
More precisely, the expectation is represented by an integration formula whose integrand is composed of weighted option prices.
For detailed information about such replication techniques, see \cite{CarrMadan} and \cite{Britten-Jones}.
Thus, one can compare the difference between the variance of return under the physical probability and the risk-neutral probability
by computing the realized variance and the synthesized option value.
Using this method, \cite{CarrWu} show that there exists a variance risk premium which implies that the risk-neutral variance is generally larger than realized variance.
The reader may refer to \cite{BakshiMadan} where the relationship between the variance risk-premium and the higher degree quantities of the return distribution is explained.
\cite{Todorov} investigate the role of jump in explaining the variance risk-premium.
See \cite{Zhang} for the relationship between the realized and risk-neutral volatilities.

The synthesized option value is often referred to as a fixed variance swap rate.
The fixed variance swap rate is the value in which investors are willing to pay to protect their wealth from variance risk.
For more information about variance swap, see \cite{DDKZ}.

We show that the risk-neutral expectations of the third and fourth moment variations are composed of synthesized option part of out-of-the-money (OTM) European options and jump correction parts.
Our empirical study shows that the option parts of expected third and fourth moment variation are good approximations to expected third and fourth moments of the return distribution, respectively.

The fact that the realized high moments variations, computable from high-frequency data, are mimicking the moments of return distribution, ones hard to compute from data, is important to hedge fat-tail risk.
To hedge the risk, we propose a new kind of variation swap.
The swap is similar with the skew swap introduced in \cite{Neuberger2012} and \cite{Kozhan},
but the floating leg is defined to be the realized third moment variation of the asset return over a fixed time period.
The third moment variation swap can be used to hedge shortfall risk of a financial asset with heavy left tail return distribution.
Under the third moment variation swap, counterparties exchange the realized third moment variation for a predetermined strike price.
The portfolio consisting of the skewed underlying asset and the third moment swap has more Gaussian-like symmetric return distribution than the original asset so that one can hedge extreme shortfall risk.

The fourth moment variation swap can be applied to an asset with leptokurtic return distribution to hedge fat-tail risk.
Similarly, the portfolio consisting of the fat-tailed underlying asset and the fourth moment variation swap has more Gaussian-like thinner tail return distribution.
We employ simulations and empirical studies to examine the performance of the variation swap.

The remainder of the paper is organized as follows.
Section~\ref{sect:moment} introduce the third and fourth moment variations.
In Section~\ref{sect:synthesizing}, we construct a mathematical framework to show that the risk-neutral expectations of third and fourth moment variations are represented by European option prices.
In Section~\ref{sect:empirical}, we present empirical studies on S\&P 500 index returns and options data.
Five-minute high frequency data of S\&P 500 index series is used to compute realized quadratic variation and covariation.
Employing some filtering methods on the index option data, we also calculate the risk-neutral expectations of quadratic variations.
In Section~\ref{sect:swap}, we explain the variation swaps to deal with tail risk and show interesting examples and empirical studies.
Section~\ref{sect:concl} concludes the paper.

\section{High order moment variations}\label{sect:moment}

Throughout this paper, we introduce a probability space with a time index set $[0,T^*]$ for some fixed $T^*>0$.
Let $(\Omega, \mathcal F, \mathbb P)$ be a complete probability space
with a filtration $\{\mathcal F_t\}_{t\in [0, T^*]}$ where $\mathcal F_{T^*} = \mathcal F$.
The measure $\mathbb P$ is the physical probability measure.
All processes introduced in this paper are defined on the probability space and those processes are adapted to the filtration.

Let $S$ denote a semimartingale asset price process and $F$ be a corresponding futures price process with maturity $0<T\leq T^*$.
Assume that there exists a risk-neutral measure $\Q$ under which every discounted asset price process is a martingale.
Also we assume $F_t = \E^{\Q} [ S_T | \mathcal F_t]$.
Define the log-return process
$$R_t = \log S_t - \log S_0.$$
We are interested in the high order moment properties of return $R_T$ over a fixed time period $[0, T]$, for example, $T=1$ or $30$ days, and the analysis are based on quadratic variation methods.

The quadratic variation process of a semimartingale $X$ is defined by
$$ [X]_t = X_t^2 - 2 \int_0^t X_u \D X_u.$$
The quadratic covariation process of $X$ and $Y$ is defined by
$$ [X,Y]_t = X_t Y_t - \int_0^t X_u \D Y_u - \int_0^t Y_u \D X_u.$$
Note that for a sequence of partition $\pi_n$ ranged over $[0,t]$, we have
$$ [X,Y]_t =  \lim_{||\pi|| \rightarrow 0} \sum_{i} ( X_{t_i} - X_{t_{i-1}} )( Y_{t_i} - Y_{t_{i-1}} ) \quad \textrm{in probablity.} $$
For the details, see~\cite{Protter}.

The realized quadratic variation of return, $[R]_t$, is an unbiased estimator of the variance of the log return $R_t$ under certain conditions \citep{ABDL}.
Thus, the realized variation become a conventional measure of the actual variance or the return.
We now define the analogous third moment covariation and fourth moment variation by
\begin{align*}
[R,R^2] \quad &\textrm{ the third moment covariation} \\
[R^2] \quad &\textrm{ the fourth moment variation}.
\end{align*}
The third moment covariation is the quadratic covariation between return and squared return processes.
The fourth moment covariation is the quadratic variation of squared return process.

It will be shown that the newly defined variations are closely related to the actual moments.
In the later, we will demonstrate that the linear transform of the expectations of the third moment covariation $[R,R^2]_T$ and fourth moment variation $[R^2]_T$ approximate the third and fourth moments of the return distribution over $[0,T]$, respectively.
Especially, such as in continuous-time stochastic volatility models, by assuming the drift term of return process is zero, the expected moment variations and actual moments are in exact linear relationship.
This is the reason that $[R,R^2]$ is called the third moment covariation and $[R^2]$ is called the fourth moment variation.

Consider a partition that $0 = t_0< \cdots < t_N =T$.
For the notational simplicity, let $R_i = R_{t_i}$.
Then we approximate the following quadratic variations and covariations by
\begin{align*}
[R]_T &\approx  \sum_{i=1}^{N} (R_{i} - R_{i-1})^2 \\
[R,R^2]_T &\approx  \sum_{i=1}^{N} (R_{i} - R_{i-1})(R^2_{i} - R^2_{i-1}) \\
[R^2]_T &\approx  \sum_{i=1}^{N} (R^2_{i} - R^2_{i-1})^2.
\end{align*}
Since $R$ and $R^2$ are semimartingales, the right hand sides of the above equations converge to the corresponding quadratic variations and covariation in probability as the mesh size of partition goes to zero.
The finite sums are called the realized (co)variations
and the realized variations are consistent estimators of the corresponding quadratic variations.
We will use the realized (co)variations as proxies for the third and fourth moment of physical distribution of $R_T$.

The each term in the finite summation of the realized variations consists of the powers of log contract.
We can rewrite
$$ (R^2_{i-1} - R^2_{i})(R_{i-1} - R_{i}) = (\Delta R_{i-1,i})^3 + 2R_{i-1} (\Delta R_{i-1,i})^2  $$
and
$$ (R^2_{i-1} - R^2_{i})^2 = (\Delta R_{i-1,i})^4 + 4R_{i-1}(\Delta R_{i-1,i})^3 + 4R_{i-1}(\Delta R_{i-1,i})^2 $$
where $\Delta R_{i-1,i} = R_i -R_{i-1}$.
The term in the third moment covariation is replicated by holding one cubic log contract and $2R_{i-1}$ square log contracts over the period $[t_{i-1}, t_{i}]$.
Similarly, the term in the fourth moment variation is replicated by holding one quartic log contract, $4R_{i-1}$ cubic log contracts and $4R_{i-1}$ square log contracts over the period $[t_{i-1}, t_{i}]$.

\begin{example}\label{Ex:Heston}
Assume that the price process follows Heston's stochastic volatility model.
Then
\begin{align*}
\D S_t &= \mu S_t \D t + \sqrt{V_t} S_t \D W_1(t)\\
\D V_t &= \kappa (\theta -V_t) \D t + \sigma\sqrt{V_t}\D W_2(t)\\
[W_1, W_2]_t &= \rho \D t.
\end{align*}
In addition,
\begin{align*}
\D R_t = \left( \mu - \frac{1}{2}V_t \right) \D t + \sqrt{V_t} \D W_1(t), \quad \D [R]_t = V_t \D t.
\end{align*}

Note that by assuming the drift in the return process is zero, we have
\begin{equation}\label{Eq:third}
[R, R^2]_T = 2\int_0^T R_t V_t \D t
\end{equation}
and
\begin{equation}\label{Eq:fourth}
[R^2]_T = 4\int_0^T R^2_t V_t \D t.
\end{equation}
In addition,
\begin{align*}
\E[R^3_T] &= 3 \E \left[ \int_0^T  R_t^2 \D R_t + 3 \E \int_0^t R_t \D [R]_t \right] \\
&= 3 \E \left[  \int_0^T R_t V_t \D t \right]
\end{align*}
and
\begin{align*}
\E[R^4_T] &= 4 \E \left[ \int_0^T  R_t^3 \D R_t \right] + 6 \E \left[ \int_0^t R^2_t \D [R]_t \right] \\
&= 6 \E \left[ \int_0^T R^2_t V_t \D t \right].
\end{align*}
By comparing above equations with Eqs.\eqref{Eq:third} and \eqref{Eq:fourth}, we conclude that in stochastic volatility models with the absence of drift in the return process,
the relation between the expected moment variations and the actual moments are linear.
More precisely,
$$\E[R^3_T] = \frac{3}{2} \E[[R^2,R]_T], \quad \E[R^4_t] = \frac{3}{2} \E [[R^2]_T].$$

To examine the bias caused by the drift term, we perform a simulation study with parameter setting $\mu = 0.05$, $\kappa = 4, \theta = 0.3, \sigma= 0.4$ and $\rho = -0.9$ with $T=1$ day.
In the left of Figure~\ref{Fig:Heston_Converge}, we plot the dynamics of the sample third moment (dotted line) and the dynamics of the sample mean of realized third moment variation $[R^2, R]_T$ multiplied by 1.5.
In spite of the drift term, the discrepancy between the limits of two quantities is relatively small.
Similarly, in the right of Figure~\ref{Fig:Heston_Converge}, the convergence of the sample fourth moment (dotted) and the sample mean of $1.5[R^2]_T$ is represented.
The sample third moment is $-5.93\times10^{-4}$ and the sample mean of $1.5[R^2,R]$ is $-6.01\times10^{-4}$.
The sample bias is about $1.4\%$.
The sample fourth moment is $2.76\times10^{-4}$ and the sample mean of $1.5[R^2]$ is $2.97\times10^{-4}$.
In this case, the sample bias is about $7.7\%$.

\begin{figure}
\begin{center}
\includegraphics[width=7.5cm]{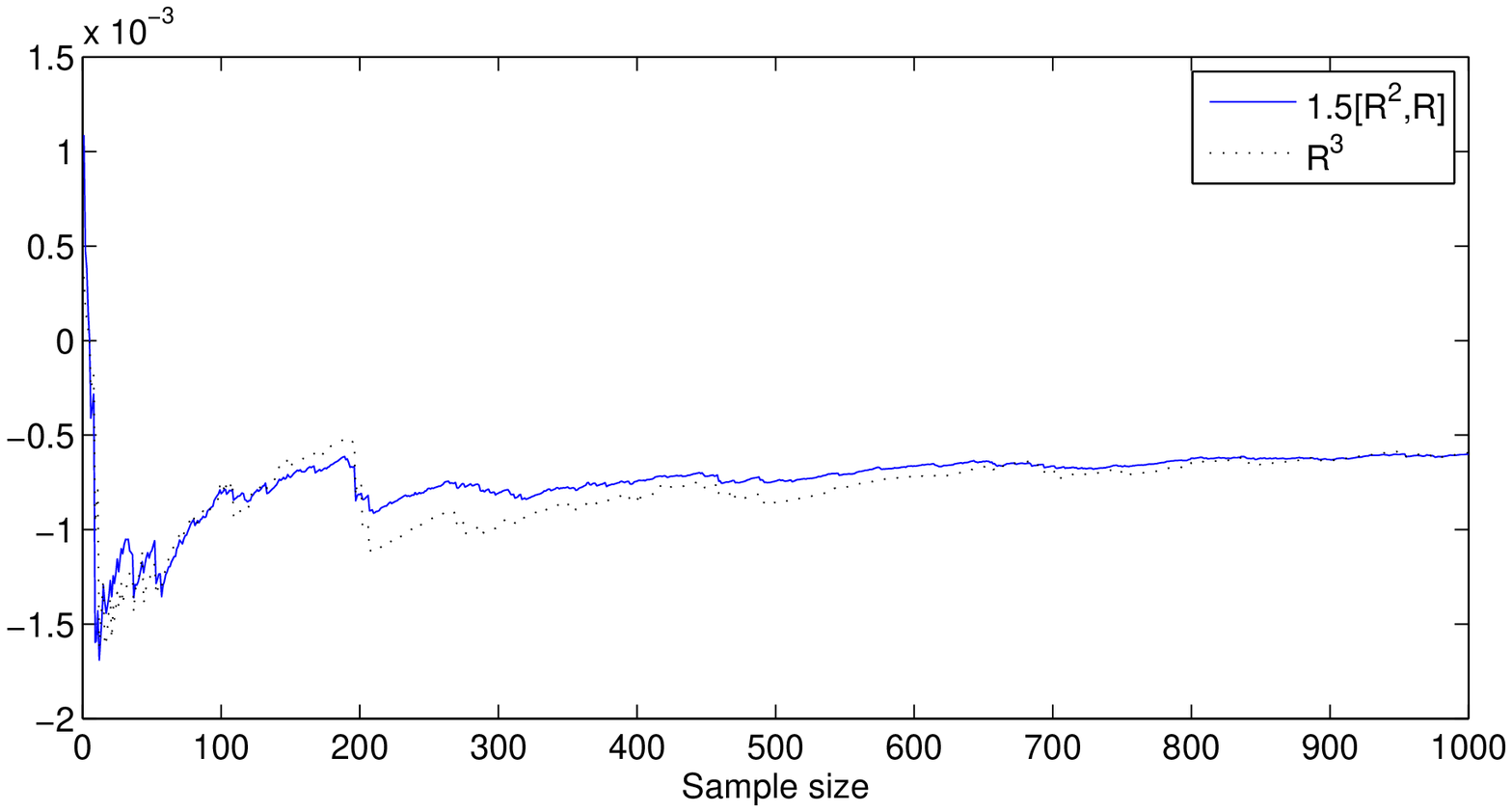}
\includegraphics[width=7.5cm]{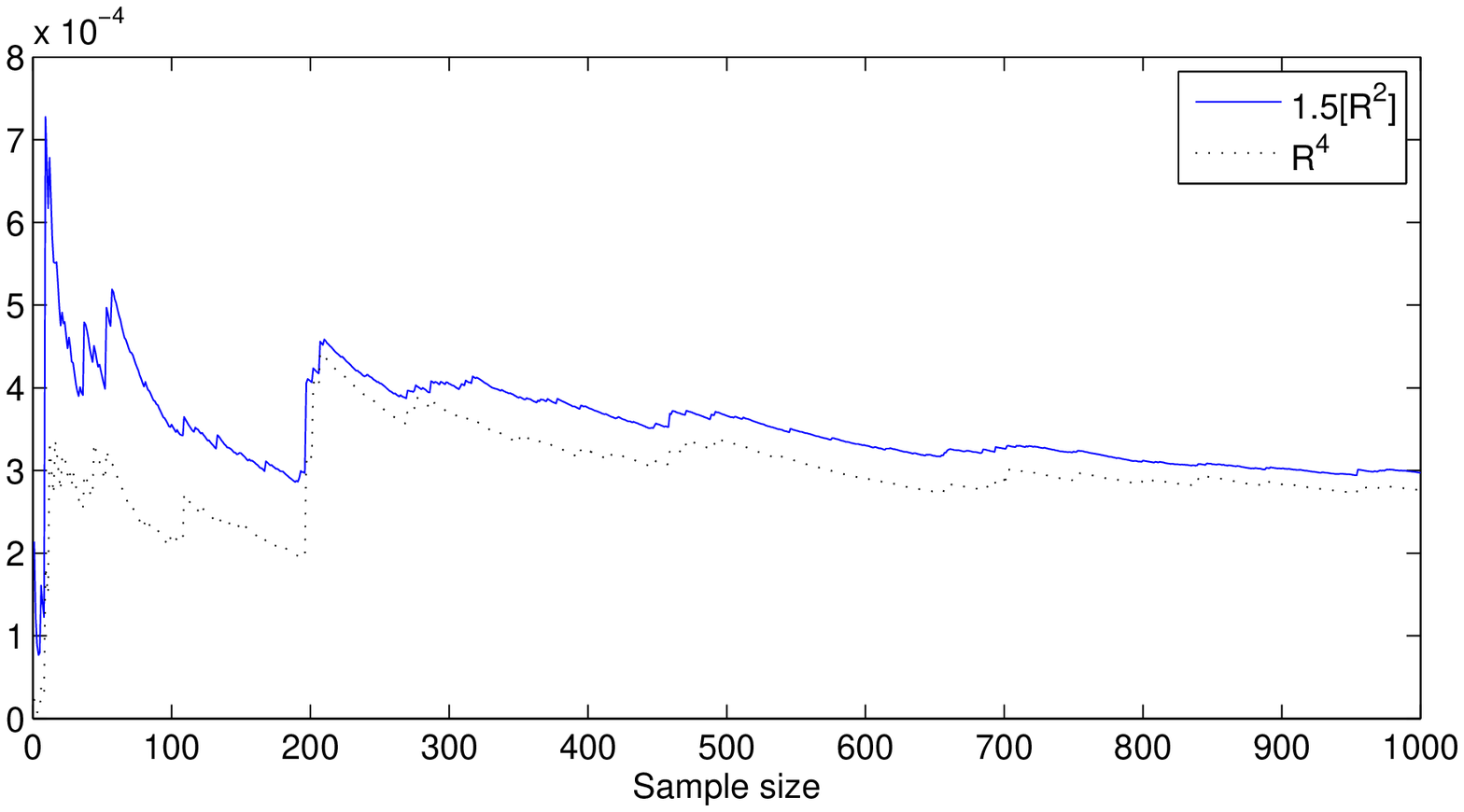}
\end{center}
\caption{Convergence of the sample third moment (dotted) and $1.5[R^2,R]$ (dashed) in the left and convergence of the sample fourth moment (dotted) and $1.5[R^2]$ (dashed) in the right}\label{Fig:Heston_Converge}
\end{figure}

\end{example}

\section{Synthesizing variations with options}\label{sect:synthesizing}

To examine the option implied third and fourth moment variations,
we derive integral formulas for the risk-neutral expectations of $[R,R^2]_T$ and $[R^2]_T$ based on European option prices.
With derived formulas, we can investigate the risk premia associated with high order moments of returns and compare the moment variations and actual moments under risk-neutral probability.

For simplicity, we assume that the instantaneous interest rate, $r$, is constant.
We define
$$
\phi\left( x, K \right)=\left\{
      \begin{array}{ll}
        p\left( x, K \right), & 0\leq K \leq \e^{rT}x, \\
        c\left( x, K \right), & \e^{rT}x < K < \infty.
      \end{array}
    \right.
$$
where $c$ and $p$ are European call and put option prices with current spot price $x$ and strike $K$, respectively.

Let $L^{2}_{\mathbb Q, [Y]}([s,t]\times\Omega)$ denote the space of adapted stochastic process $X$ such that
$$\mathbb E^{\Q} \left[ \left. \int_{s}^{t}  X^2_u \D [Y]_u  \right| \F_s \right] < \infty. \quad \textrm{a.s.}$$
Under the condition, we guarantee that the stochastic integral of $X$ with respect to a $\Q$-martingale $Y$ is a $\Q$-martingale.
For the detailed information, consult~\cite{Kuo}.

To derive the integral formula, we need the following technical conditions:
\begin{equation}\label{Eq:condition}
\frac{1}{F}, \frac{R+1}{F}, \frac{R^2+ 2R +2}{F} \in L^{2}_{\mathbb Q}([0,T]_{[F]}\times\Omega)
\end{equation}
In the next theorem, the risk-neutral expectations of the variations are represented by the sums of integral whose integrands are weighted European option prices and jump correction parts.
For the risk-neutral expectations of the high moment variations, we have analogous results with \cite{CarrWu}.

\begin{theorem}\label{Thm:representation}
Under the condition~\eqref{Eq:condition}, the risk-neutral expectations of variations are represented as sums of synthesized option prices and jump correction terms. More precisely,
\begin{align}
\E^{\Q}\left[\, [R]_T \right] ={}& 2e^{rT}  \int_{0}^{\infty} \frac{1}{K^2} \phi(S_0, K) \D K + J_2\\
\E^{\Q}\left[\, [R, R^2]_T \right]  ={}& 4e^{rT} \int_{0}^{\infty} \left(\log\frac{K}{S_0}\right) \frac{1}{K^2} \phi(S_0, K)  \D K + J_3\label{Eq:tm}\\
\E^{\Q}\left[\, [R^2]_T \right] ={}& 8e^{rT} \int_{0}^{\infty} \left(\log\frac{K}{S_0}\right)^2 \frac{1}{K^2} \phi(S_0, K)  \D K +J_4.\label{Eq:fm}
\end{align}
The jump correction terms $J_2, J_3$ and $J_4$ are of order $O((\Delta R_t)^3)$.
\end{theorem}

\begin{proof}
See Appendix~\ref{Apped:proof}.
\end{proof}

\section{Empirical analysis}\label{sect:empirical}

For the empirical study, we use the historical index and option data of S\&P 500.
The period of data set ranges from January 2001 to December 2007.
We only use Wednesday option prices to avoid weekend effects and among those the closing quotes of the days are selected.
We exclude options with too short time-to-maturity, less than 10 days.
We only use OTM option data since OTM options are more liquid than ITM options.

In Table~\ref{Table:data} we summarize the statistics of total 26,227 option data by time-to-maturity and moneyness, $K/S$,
where $K$ is strike price and $S$ is S\&P 500 index.
In the table we also report Black-Scholes implied volatilities and sample size.
In panel A we have put option data with total 15,814 observations and panel B is for 10,413 call option data.

\begin{table}
\caption{S\&P 500 option data for time to maturity $T < 40$ and $40 \leq T < 70$ days}\label{Table:data}
\begin{center}
\begin{tabular}{llrrrrr}
\toprule
$K/S$ & &\multicolumn{2}{c}{$T < 40$ } & \multicolumn{2}{c}{$40 \leq T < 70$ }  \\
\cmidrule(rl){3-4} \cmidrule(rl){5-6}
 & & Mean & SD & Mean & SD \\
\midrule
$<0.85$ & Price & 0.60 &  0.93 & 1.42 &  1.79\\
& Implied vol. &  0.39 &  0.10& 0.33 & 0.08 \\
& Sample size &  1797& &1567  \\
\midrule
$0.85-1.00$ & Price  & 5.16 & 6.19 & 9.47 &8.96 \\
& Implied vol. & 0.22 &  0.08& 0.21& 0.06 \\
& Sample size &  8529& & 3921 \\
\midrule
$1.00-1.15$& Price & 8.50& 9.13& 5.61& 7.06\\
& Implied vol. &  0.13 & 0.05& 0.16&  0.06\\
& Sample size & 3429& &6513 \\
\midrule
$>1.15$ & Price & 0.61 & 2.85& 1.03& 1.08\\
& Implied vol. & 0.29 & 0.09 & 0.22 & 0.04\\
& Sample size & 258& & 213 \\
\bottomrule
\end{tabular}
\end{center}
\end{table}

In Table~\ref{Table:summary} we summarize the statistics for the realized variations computed from S\&P 500 and the option implied expectations of variations with $T = 30$ days.
Both are annualized.
The means, standard deviations, skewness, kurtosis and Ljung-Box test statistics with 18 lags are reported.
Ljung-Box test examine whether the autocorrelation of a time series are different from zero.
Though the test is known to be valid under the strong white noise assumption, we use the test statistics to quantify the serial conditional correlations.
The statistics suggest that the risk-neutral variations are more serially correlated than realized ones.
Note that the absolute values of sample means of all option implied expectations of quadratic variations are higher than average realized variations.
The realized and option implied variations for variance are persistence.
The realized and option implied third moment covariation and fourth moment variations are less persistence.

\begin{table}
\caption{Statistics for the realized variations and the option implied expectations of variations}\label{Table:summary}
\begin{center}
\begin{tabular}{lrcrrr}
\toprule
Variation & Mean & SD & Skew & Kurt. & Ljung-Box\\
\midrule
\multicolumn{4}{l}{\it Panel A: Realized variation}\\
\noalign{\smallskip}
$[R,R]_T$ & $2.27\times10^{-2}$ & $2.27\times10^{-2}$ & $2.53$ & $10.80$ & 2094.1\\
$[R,R^2]_T$ & $-8.49\times10^{-4}$ & $3.80\times10^{-3}$ & $-3.27$ & 21.42 & 313.6\\
$[R^2,R^2]_T$ & $2.65\times10^{-4}$ & $8.59\times10^{-4}$ & 5.51 & 38.92 & 410.7\\
\midrule
\multicolumn{4}{l}{\it Panel B: Option implied expectation}\\
\noalign{\smallskip}
$[R,R]_T$  & $4.45\times10^{-2}$ & $3.63\times10^{-2}$ & 1.74 & $6.15$ & 2687.4\\
$[R,R^2]_T$ & $-3.23\times10^{-3}$ & $4.28\times10^{-3}$ & $-2.96$& 14.21 & 1680.7\\
$[R^2,R^2]_T$ & $1.48\times10^{-3}$ & $2.33\times10^{-3}$ & 3.54 & 18.59 & 1556.6\\
\bottomrule
\end{tabular}
\end{center}
\end{table}

For the quadratic variation $[R,R]_T$, the estimation of return variance, the results are similar with \cite{CarrWu}.
Both realized and option implied covariations for the third moment have negative values of sample mean and
the sample mean of option implied covariation is smaller than realized one.
This implies that the risk-neutral distribution of return is more negatively skewed than physical distribution.
Also note that the sample mean of the fourth moment variation is greater than realized one and this means the risk-neutral return distribution has fatter tail.

In Figure~\ref{Fig:third} the dynamics of realized third moment covariation $[R,R^2]$ (solid line) and its option implied expectation (dash-dot line) are plotted.
The third moment covariations are generally negative and the absolute values of option implied expectations are larger than those of realized covariations.
This result is agree with the fact that return distribution is left skewed and risk-neutral distribution is more heavily left skewed.
From 2001 to 2002 (during dot-com bubble) and after 2007 (at the beginning of financial crisis), the third moment covariations are more fluctuating and have large absolute values in general.
Otherwise, the absolute values of covariations are closed to zero.

In Figure~\ref{Fig:fourth} the historical behavior of realized fourth moment variation $[R^2]$ (solid line) and its option implied expectation (dash-dot line) are plotted.
Similarly with the previous one, option implied variations are generally higher than realized variations.
and from 2001 to 2002 and after 2007, the variations are large.
Otherwise, the values are close to zero.

\begin{figure}
\begin{center}
\includegraphics[width=10cm]{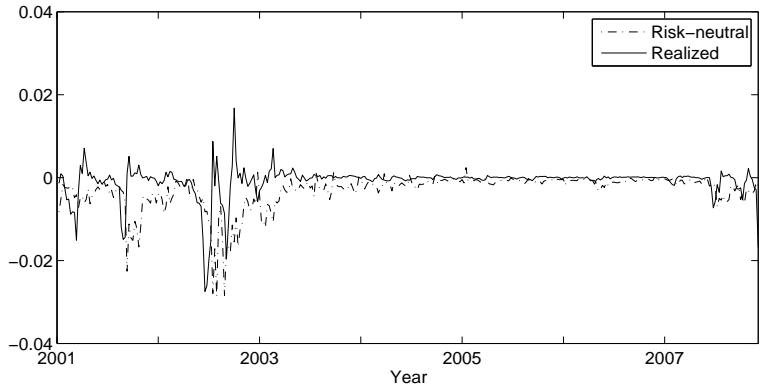}
\end{center}
\caption{The dynamics of realized (solid) and option-implied (dash-dot) third moment variation}
\label{Fig:third}
\end{figure}

\begin{figure}
\begin{center}
\includegraphics[width=10cm]{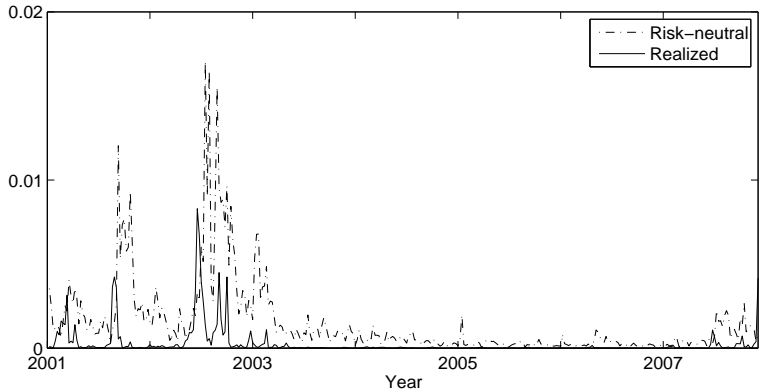}
\end{center}
\caption{The dynamics of realized (solid) and option-implied (dash-dot) fourth moment variation}
\label{Fig:fourth}
\end{figure}

In Example~\ref{Ex:Heston}, we show that when the price process follows a stochastic volatility model, the expected variations and actual moments are in linear relationship under the absence of the drift in return process.
However, the exact relation between variations and moments may depend on the choice of model as in
\cite{DuKapadia} where the model-specific mathematical relation between the quadratic variation and the variance are derived.

It is worthwhile to empirically examine the relationship.
Since the actual moments of return distribution is hard to compute under physical probability, we employ the analysis under risk-neutral probability.

In \cite{BakshiKapadiaMadan}, the cubic and quartic contracts are defined to have the payoffs $R_T^3$ and $R_T^4$, respectively.
Then the risk-neutral expectations of the contracts, i.e., risk-neutral third and fourth moments of return $R_T$ are represented by European option prices with maturity $T$.
More precisely,
\begin{align}
\E^{\Q}[R_T^3] &= e^{rT}\int_{0}^{\infty}\frac{6\log(K/S_0)-3(\log(K/S_0))^2}{K^2}\phi(K) \D K \label{Eq:cubic}\\
\E^{\Q}[R_T^4] &= e^{rT}\int_{0}^{\infty}\frac{12(\log(K/S_0))^2-4(\log(K/S_0))^3}{K^2}\phi(K) \D K \label{Eq:quartic}.
\end{align}

The empirical study shows that the risk-neutral expectation of the third moment covariation $\E^{\Q}[[R,R^2]_T]$ is closely related to $\E^{\Q}[R_T^3]$.
Similarly, the risk-neutral expectation of the fourth moment variation $\E^{\Q}[[R^2]_T]$ has a significant relationship with $\E^{\Q}[R_T^4]$.
In Figure~\ref{Fig:cubic} we compare the dynamics of the option part of $\E^{\Q}[[R,R^2]_T]$ (dash-dot) with the dynamics of the risk-neutral expectation of third moment (solid).
The period time for the variations is fixed to 30 calender days.
Also, in Figure~\ref{Fig:quartic} we compare the dynamics of the option part of $\E^{\Q}[[R^2]_T]$ (dash-dot) with the dynamics of risk-neutral expectation of fourth moment (solid).

\begin{figure}
\begin{center}
\includegraphics[width=10cm]{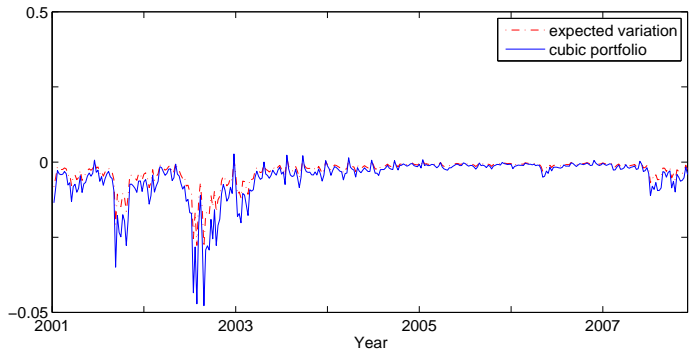}
\end{center}
\caption{The dynamics of risk-neutral expectations of third moment (solid) and the synthesized option values for $\E^{\Q}[[R,R^2]_T]$ (dash-dot) computed from S\&P 500 European option prices}
\label{Fig:cubic}
\end{figure}

\begin{figure}
\begin{center}
\includegraphics[width=10cm]{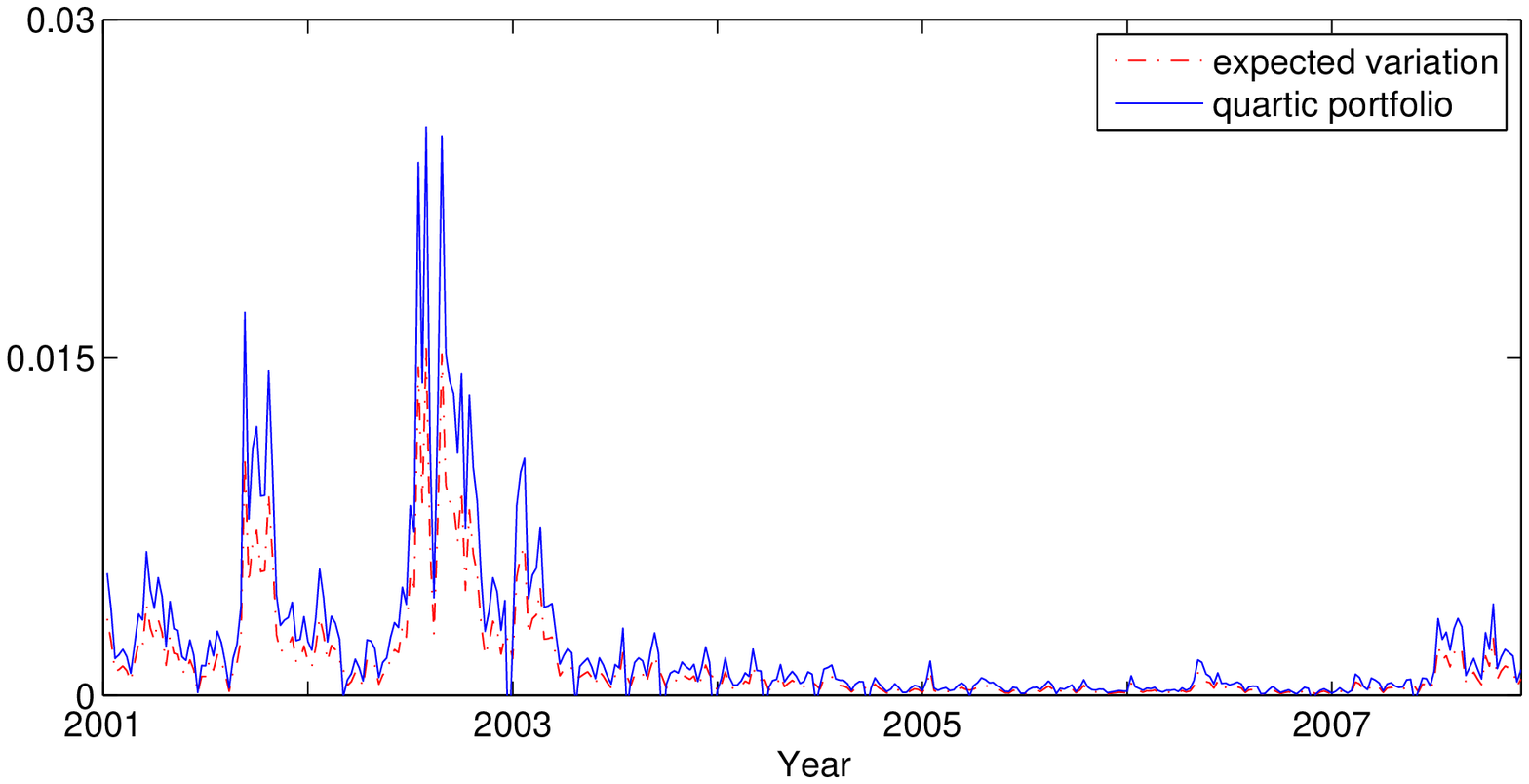}
\caption{The dynamics of risk-neutral expectations of fourth moment (solid) and the synthesized option values for $\E^{\Q}[[R^2]_T]$ (dash-dot) computed from S\&P 500 European option prices}
\end{center}
\label{Fig:quartic}
\end{figure}

Based on the empirical study, we may surmise the relations from empirical observations.
We employ the following linear regressions to examine the relationship between variations and moments:
\begin{align*}
\E^{\Q}[R_T^3] &= \beta_0 + \beta_1 \E^{\Q} [[R,R^2]_T] \\
\E^{\Q}[R_T^4] &= \beta_0 + \beta_1 \E^{\Q} [[R^2,R^2]_T]
\end{align*}
The results are shown in Table~\ref{Table:variation_moment}.
The regression graphs are plotted in Figure~\ref{Fig:regression}.
In the left of the figure, the estimated linear regression of the risk-neutral third moment vs. $[R^2,R]$ is plotted and in the right, the regression of the risk-neutral fourth moment vs. $[R^2]$ is plotted.

\begin{table}
\centering
\caption{Linear relationship between variations and moments}\label{Table:variation_moment}
\begin{tabular}{ccccccc}
\hline
& $\beta_0$ & (s.e.) & $\beta_1$ & (s.e.) & adjusted $\mathrm{R}^2$ & RMSE\\
\hline
third moment & 0.000 & (0.000) & 1.698 & (0.002) & 1.000 & 0.000\\
fourth moment & $0.000$ & (0.000) & 1.623 & (0.002) & 1.000 & 0.000\\
\hline
\end{tabular}
\end{table}

\begin{figure}
\centering
\includegraphics[width=7cm]{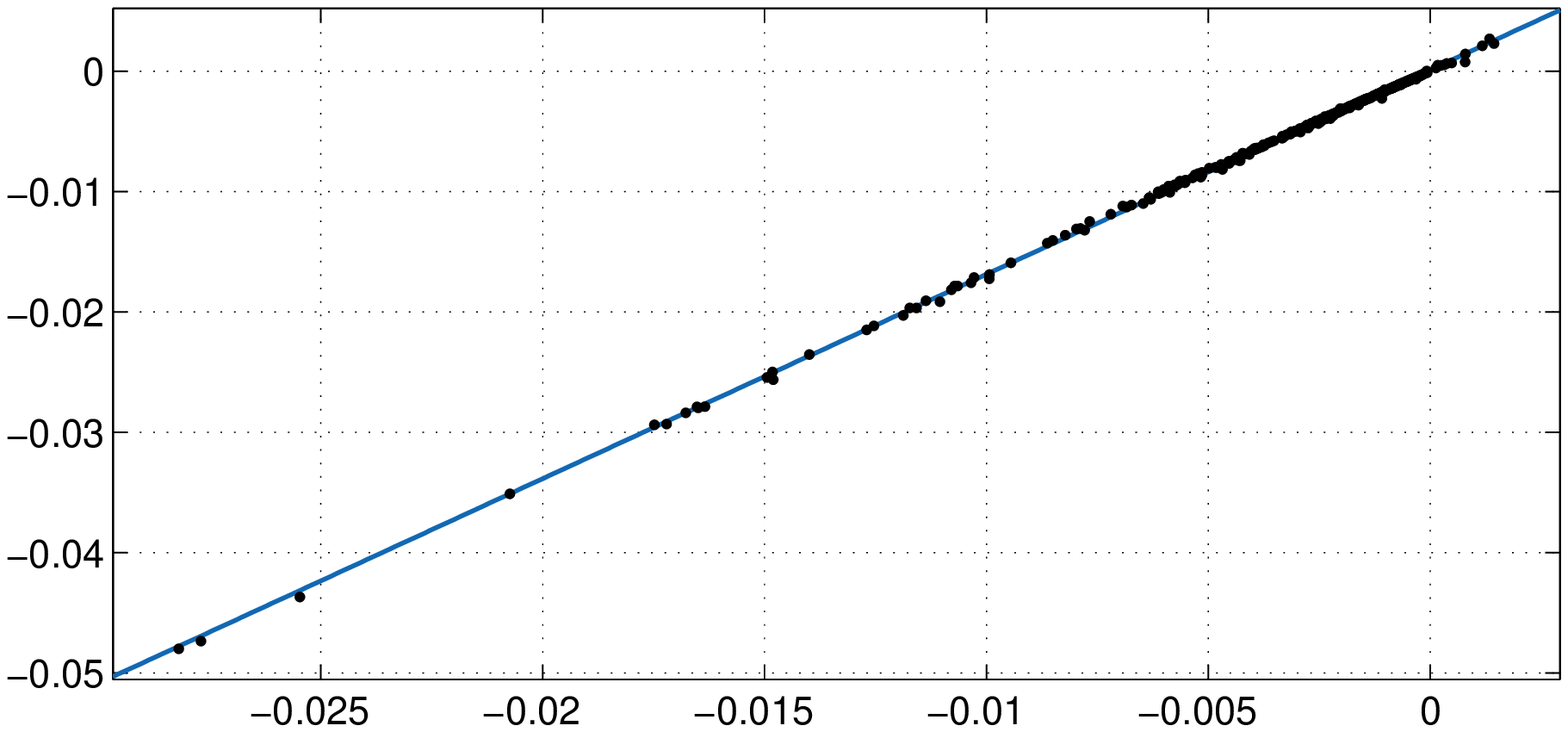}
\includegraphics[width=7cm]{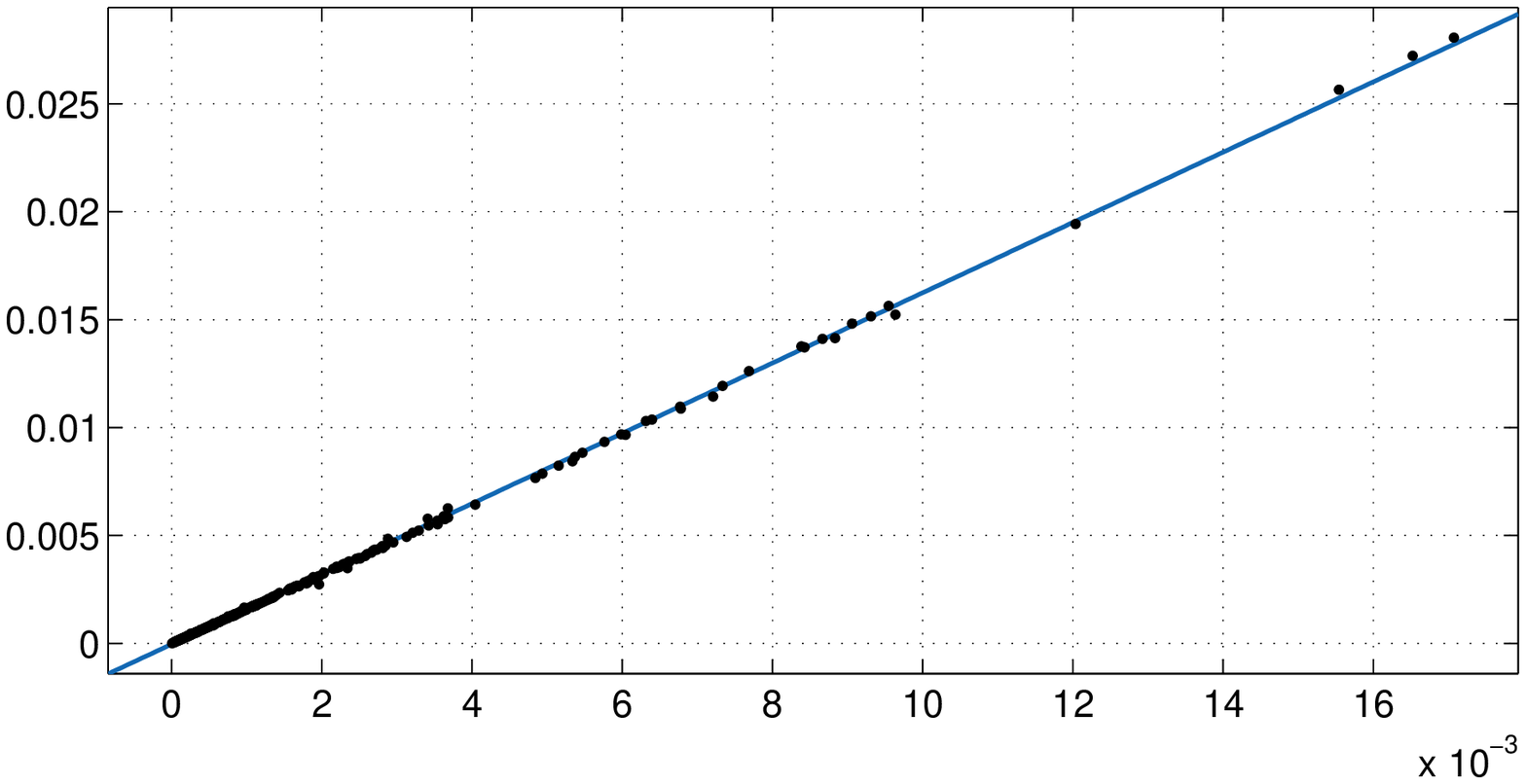}
\caption{Linear regression of risk-neutral third moment vs. $[R^2,R]$ (left) and risk-neutral fourth moment vs. $[R^2]$}\label{Fig:regression}
\end{figure}

Finite sample behaviors of errors are reported in Table~\ref{Table:error} where we use the variations as approximations to moments under risk-neutral probability.
Means of errors and 0.025 and 0.975 quantiles are presented.

\begin{table}
\centering
\caption{Finite sample behaviors of errors with mean and quantiles}\label{Table:error}
\begin{tabular}{cccc}
\hline
 & Mean & 0.025 & 0.075 \\
\hline
$\E^{\Q}[R_T^3] - \beta_1 \E^{\Q}[[R,R^2]_T]$ & $1.16\times10^{-4}$ & $-1.56\times10^{-4}$ & $3.84\times10^{-4}$\\
$\E^{\Q}[R_T^4] - \beta_1 \E^{\Q}[[R^2,R^2]_T]$ & $-3.29\times10^{-5}$ & $-1.83\times10^{-4}$ & $6.63\times10^{-5}$ \\
\hline
\end{tabular}
\end{table}

In the left of Figure~\ref{Fig:integrands} the integrands for the third moment variation in Eq.~\eqref{Eq:tm} (dash-dot) and the cubic portfolio in Eq.~\eqref{Eq:cubic} (solid) are plotted where $S_0 = 100$.
For the integrand in Eq.~\eqref{Eq:tm}, we multiply 1.698 (which comes from the linear regression in Table~\ref{Table:variation_moment}) for the linear approximation.
In the right of Figure~\ref{Fig:integrands} the integrands for the fourth moment variation in Eq.~\eqref{Eq:fm} (dash-dot) and the quartic portfolio in Eq.~\eqref{Eq:quartic} (solid) are compared.
For the integrand in Eq.~\eqref{Eq:fm}, we multiply 1.621 for the linear approximation.

\begin{figure}
\begin{center}
\subfigure[]{
\includegraphics[width=5cm]{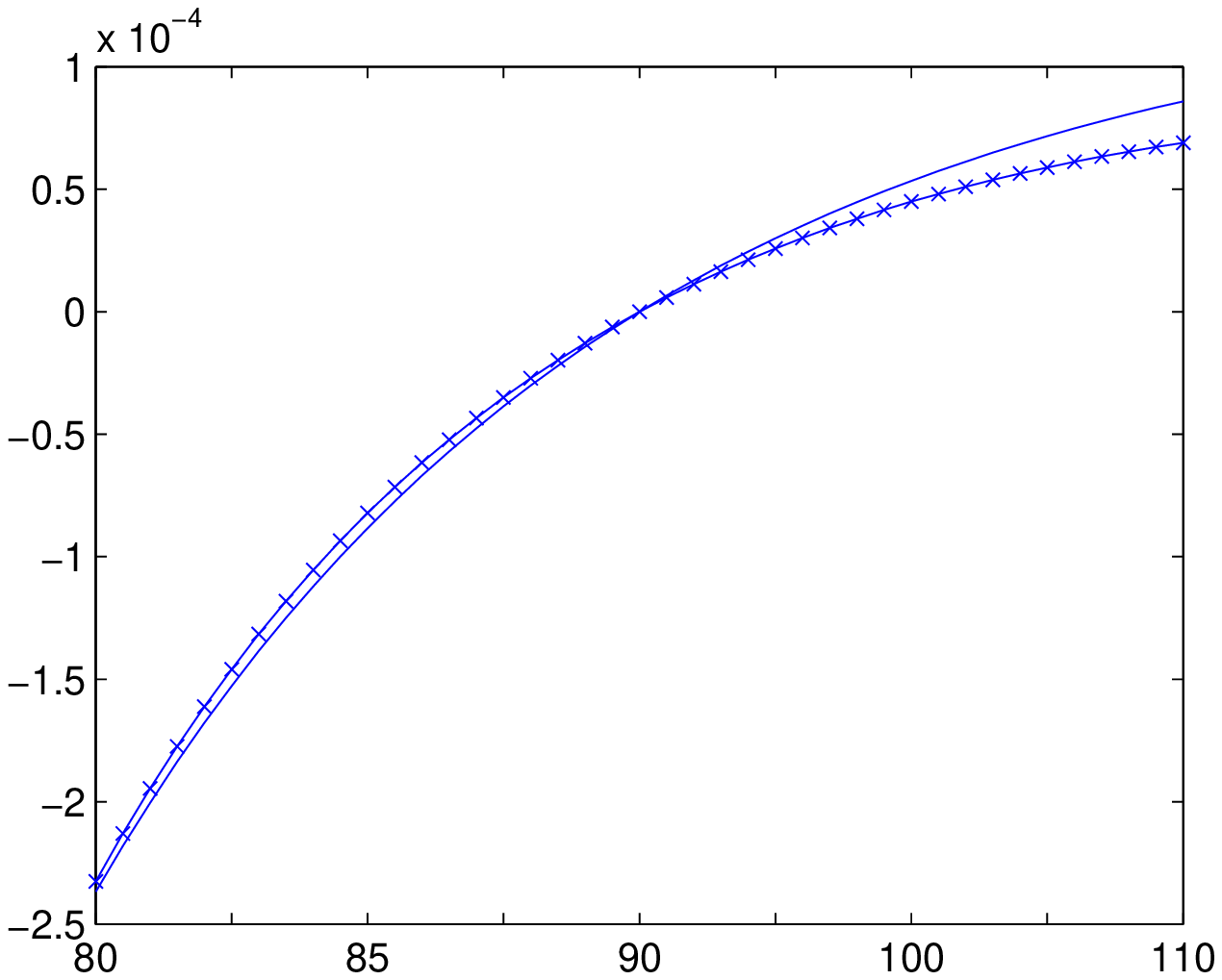}
}
\quad
\subfigure[]{
\includegraphics[width=5cm]{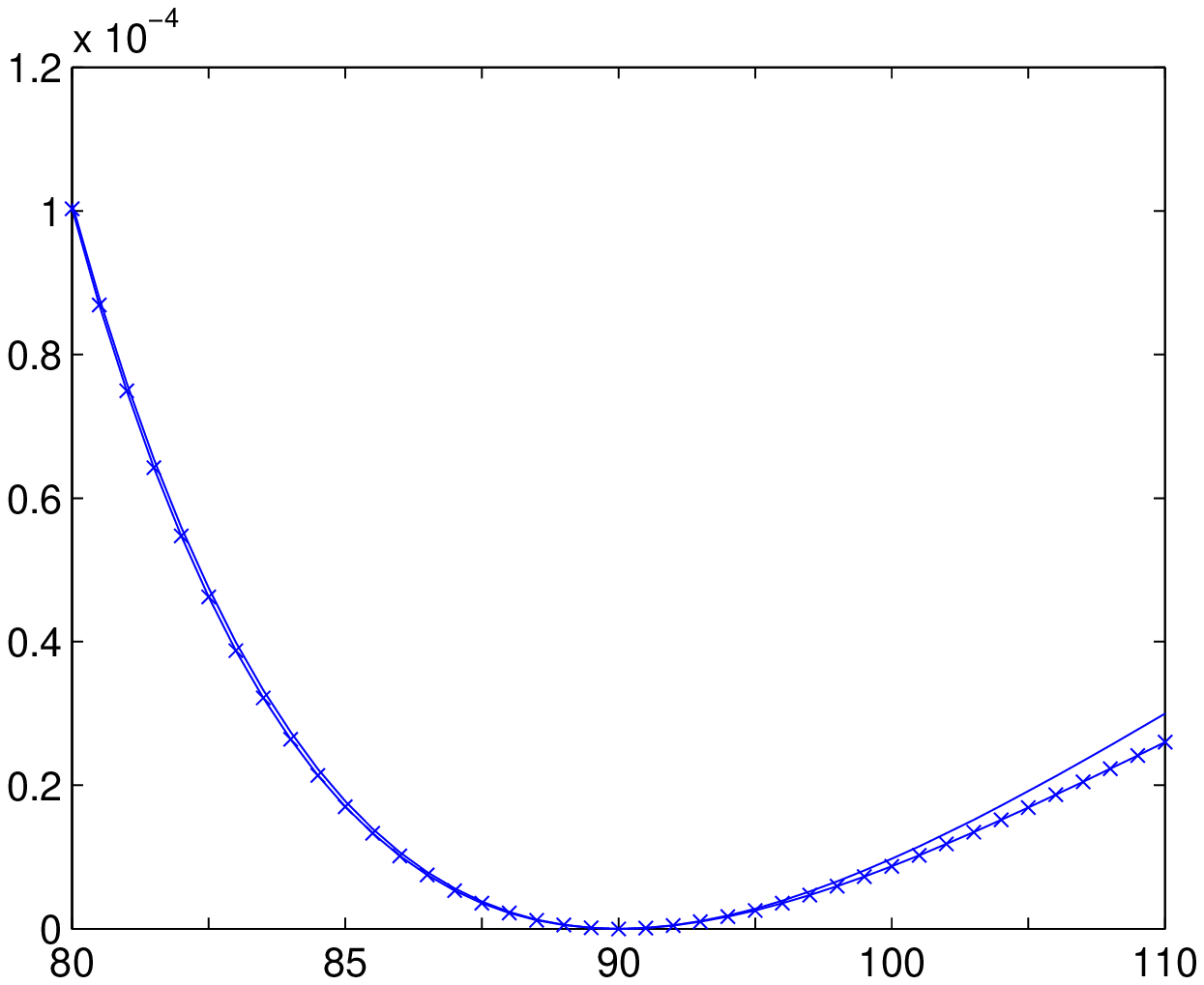}
}
\end{center}
\label{Fig:integrands}
\caption{(a) the integrands for the third moment variation (dash-dot) and the cubic portfolio (solid), (b) the integrands for the fourth moment variation (dash-dot) and the quartic portfolio (solid)}
\end{figure}

Now we compare risk-neutral skewness and kurtosis with approximations to the quantities based on our method.
The risk-neutral skewness and kurtosis is calculated based on the integration formulas of expected risk-neutral third and fourth moments by the method in \cite{BakshiKapadiaMadan}.
Also using the third moment covariation and the fourth moment variation combined with linear relationship in Table~\ref{Table:variation_moment} as approximations to third and fourth moments, we compute approximations to risk-neutral skewness and kurtosis.
Then we compare the results.
In Table~\ref{Fig:skewness} the dynamics of risk-neutral skewness (solid) and its approximation (dash-dot) based on the variation method are illustrated.
In Table~\ref{Fig:kurtosis} the dynamics of risk-neutral kurtosis (solid) and its approximation (dash-dot) based on the variation method are presented.
The figures suggest that our approximations are very close to the skewness and kurtosis.

\begin{figure}
\begin{center}
\includegraphics[width=10cm]{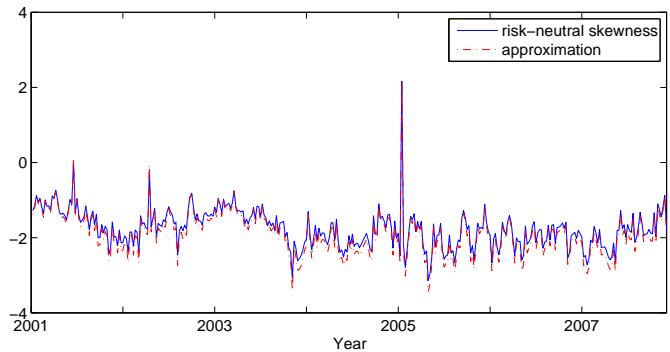}
\end{center}
\caption{The dynamics of risk-neutral skewness (solid) and its approximation (dash-dot)}
\label{Fig:skewness}
\end{figure}

\begin{figure}
\begin{center}
\includegraphics[width=10cm]{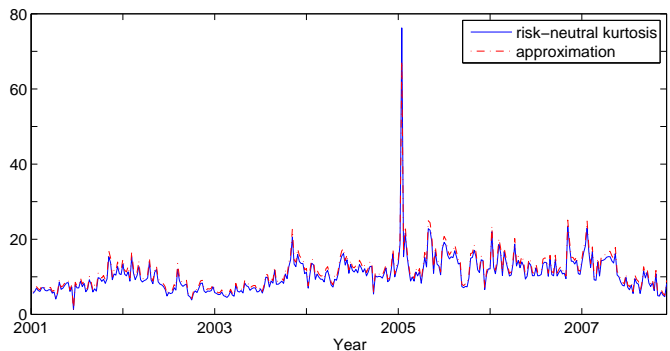}
\end{center}
\caption{The dynamics of risk-neutral kurtosis (solid) and its approximation (dash-dot)}
\label{Fig:kurtosis}
\end{figure} 

\section{Swap for variation}\label{sect:swap}
The swap for the third moment variation can be used to hedge shortfall risk of a financial asset.
The trading mechanism for the third moment variation swap is similar with the existing variance swap.
The one leg of the swap is the floating number based on the realized third moment variation of the asset return over a fixed time period and the other leg of the swap is a predetermined fixed strike price.
One different thing from the variance swap is that the floating leg of the third moment variation swap can be negative value.
An investor can hedge the downfall risk by buying a OTM put option, but in some option market the OTM put option is relatively highly priced by the speculators who seek high return with small amount of initial capital.
In this situation, the swap for third moment variation could be an alternative choice of over-the-counter product to hedge the risk.

One of the important and interesting features to use the third moment variation swap is that it is able to construct a portfolio whose return distribution is symmetric and Gaussian-like although the return distribution of the underlying asset is asymmetric and has left heavy tail.
Note that financial asset return distributions are generally negatively skewed.

To show this, we employ a simulation study to generate a negatively skewed return distribution.
Consider the Heston model in Example~\ref{Ex:Heston} with coefficients $\mu = 0.05, \kappa = 4,  \theta =0.09, \sigma= 0.4$ and $\rho = -0.9$.
Because of the negative $\rho$, the return distribution generated by the model is negatively skewed.
We simulate $10^5$ paths with $T=30$ days.
The skewness of the sample distribution is $-0.5030$ and
the QQ-plot for the return $R_T$ versus the standard normal distribution is in the left of Figure~\ref{Fig:QQ3}.
\begin{figure}
\begin{center}
\includegraphics[width=6cm]{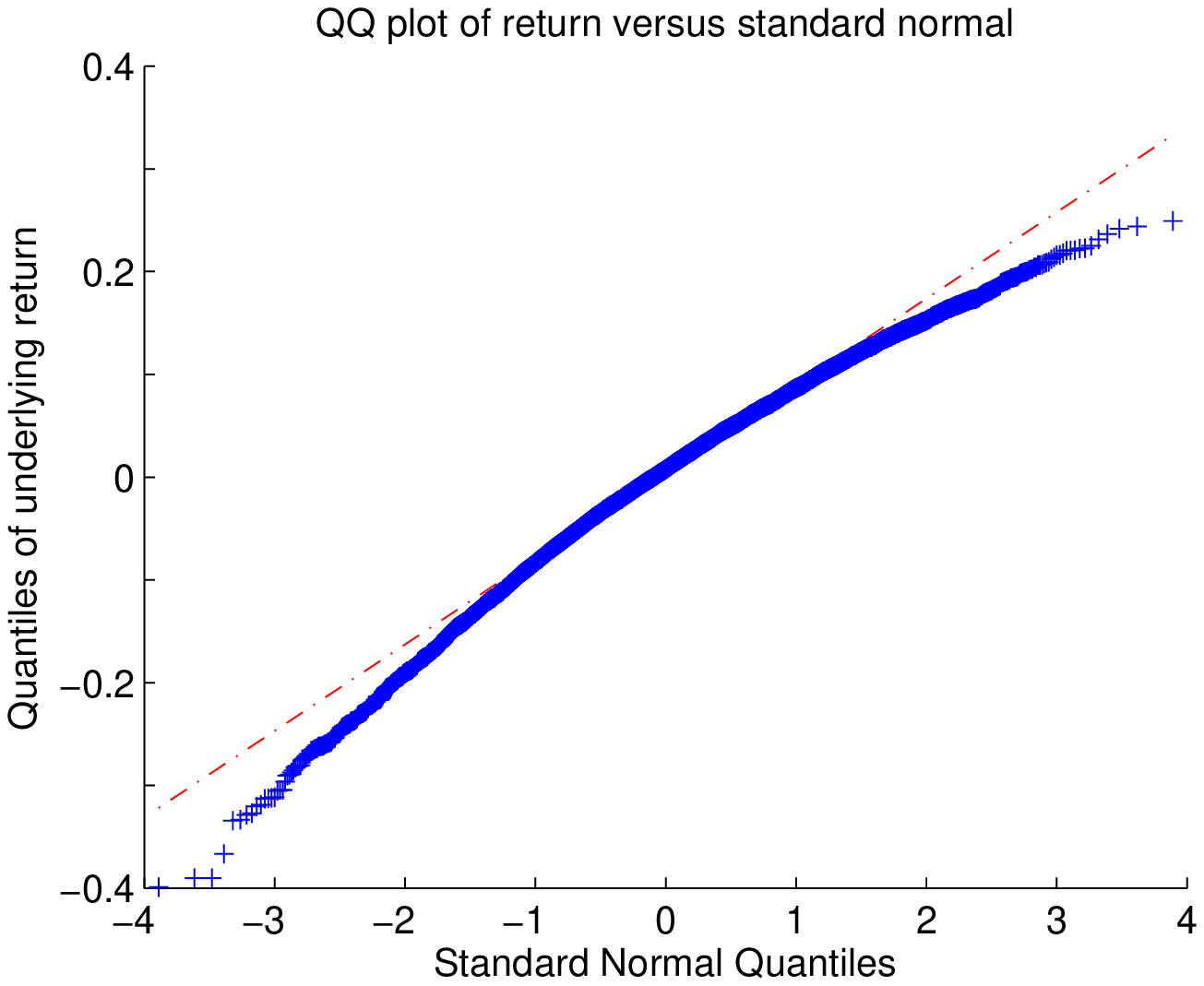}\quad
\includegraphics[width=6cm]{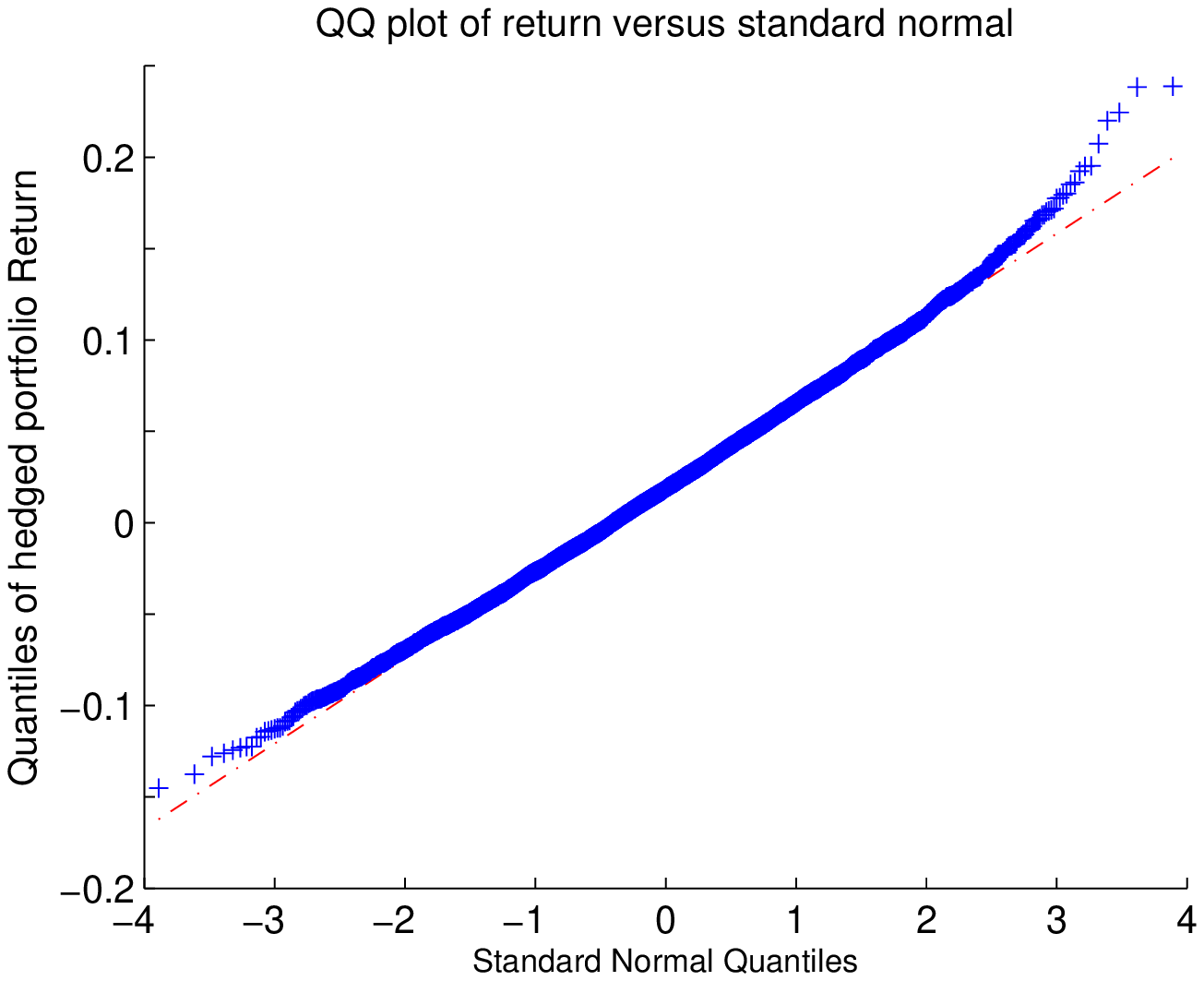}
\end{center}
\caption{QQ-plots for skewed return distribution versus normal (left) and the return distribution of hedged portfolio versus normal (right) under the Heston model}
\label{Fig:QQ3}
\end{figure}

Now suppose an investor holding an asset whose price process is assumed to follow the given Heston's model want to hedge negative tail risk so that the return distribution of the portfolio become Gaussian-like distribution by contracting a third moment variation swap.
A simple way to find an appropriate number of hedge position is to use linear regression.

Consider that an investor holding an asset $S$ buy $\beta S_0$ numbers of floating leg of the third moment variation $[R,R^2]$ over a time period $[0,T]$.
For simplicity, the fixed leg of the swap is zero and the exchange of floating leg occur at time $t$.
Then, the cash flow of the investor at time $T$ is $S_T - \beta S_0 [R,R^2]_T$.
Note that the log-return of the portfolio at time $T$ is
$$ \log \left(\frac{S_T - \beta S_0[R,R^2]_T }{S_0} \right) \approx R_t - \beta [R,R^2]_t.$$

Linear regression is applied to find the coefficients that fits the $R_t$, in a least-squares sense against the realized third moment variation $[R,R^2]_T$.
In this example, the result is
$$R_T = 0.0163 + 85.4949 [R,R^2]_T.$$
(Note that this is a simple example to find the number of hedge position and we cannot say this is the best method to determine the number of hedge position. However, the result based on this simple method is quiet good.)
Therefore, one construct a portfolio consisting of one stock and short positions of 85.4949 third moment swap
when we set $S_0=1$ for simplicity.
In other words, one receive at maturity $S_T - 85.4949 [R,R^2]_T$.

The QQ-plot for the return of the portfolio versus a normal distribution is in the right of Figure~\ref{Fig:QQ3}.
If the distribution of the hedged portfolio is close to normal, the plot will be close to linear.
The plot shows that the hedged portfolio has more Gaussian-like return distribution than the underlying asset and robust to shortfall risk.

Also the swap for the fourth moment variation can be used to hedge tail risk when the underlying return distribution is leptokurtic.
Consider the Heston model with coefficients $\mu = 0.05, \kappa = 0.5,  \theta =0.09, \sigma= 1.2$ and $\rho = 0$.
With these setting of parameters, the return distribution have fat tail as in the left of Figure~\ref{Fig:QQ4} (with sample kurtosis 4.8273).
In this situation, one receive the floating leg of the realized fourth moment variation if the underlying return is negative and pay the realized fourth moment variation if the underlying return is positive.
To find an appropriate number of the contract of the swap, we apply linear regression for the absolute value of return against the realized fourth moment variation:
$$|R_T| = 0.0017 + 53.9310 [R^2]_T$$
(Similarly with the previous example, we cannot say this is the best method to determine the number of hedge position.)
The return of portfolio consisting of the underling and 53.9310 fourth moment swap is more Gaussian-like distribution as in the right of Figure~\ref{Fig:QQ4}.
\begin{figure}
\begin{center}
\includegraphics[width=6cm]{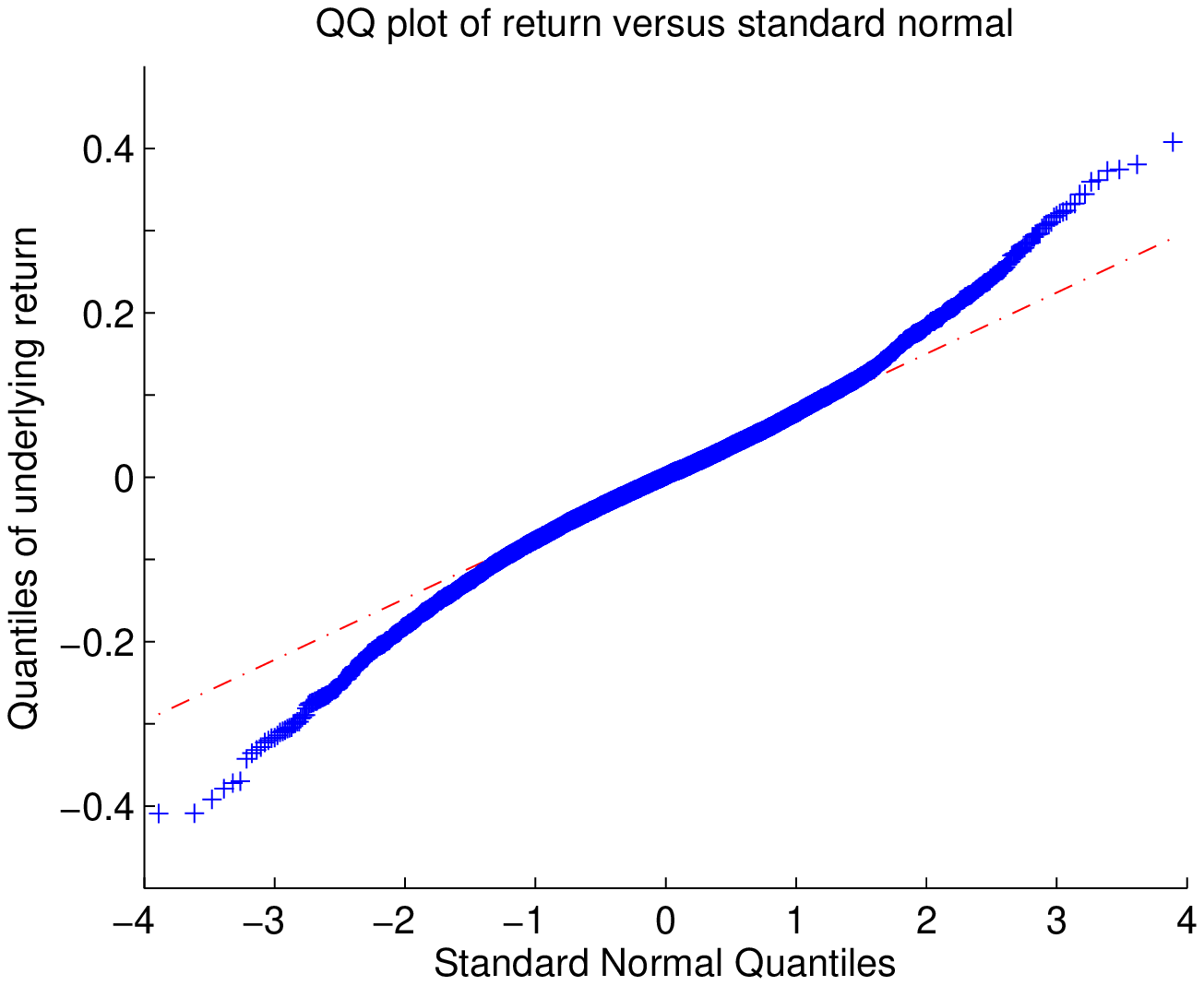}\quad
\includegraphics[width=6cm]{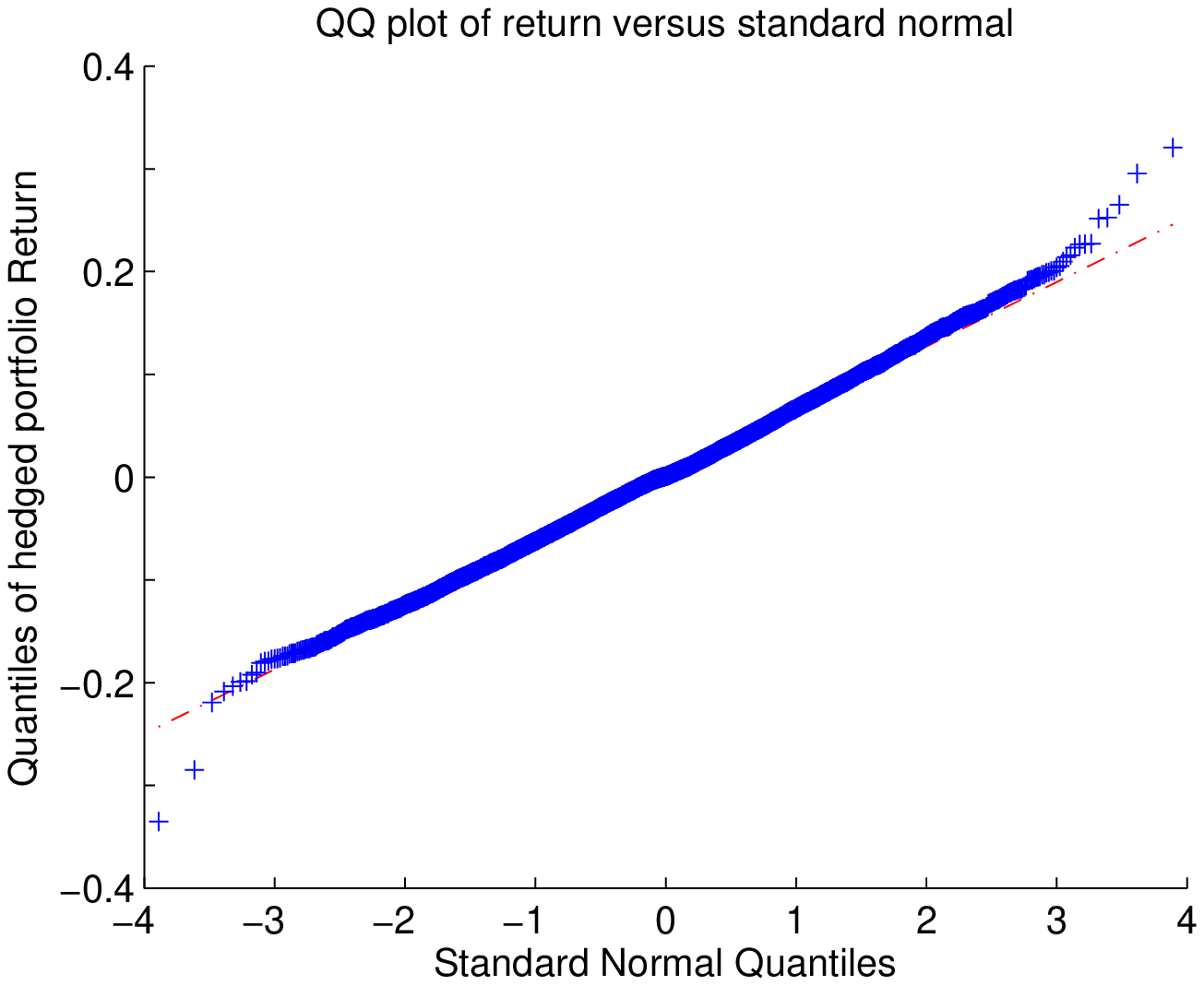}
\end{center}
\caption{QQ-plots for fat-tailed return distribution versus normal (left) and the return distribution of hedged portfolio versus normal (right) under the Heston model}
\label{Fig:QQ4}
\end{figure}

To examine empirical application, five-minute time series of S\&P 500 index ranged from January 2001 to December 2007 are used.
We construct four kinds of portfolio to test the tail behaviors.
First one is S\&P 500 index over one month period.
Second one is the portfolio composed of the index and the floating leg of third moment variation over one month period.
Third one is the portfolio composed of the index and the floating leg of fourth moment variation as explained in the simulation study.
The last one is the portfolio composed of the index and the realized third moment defined by \cite{Neuberger2012} and \cite{Kozhan}.

Keeping in mind that our goal is to construct a portfolio whose return distribution is more Gaussian-like,
the error is defined as the the difference between the quantile of the constructed portfolio return and the quatile of the normal distribution.
The mean and variance of the normal distribution is set to match to the mean and variance of constructed portfolio.
The weight of portfolio is determined to minimize the root mean square errors.
The results are presented in Figure~\ref{Fig:QQ_SP}.

\begin{figure}
\begin{center}
\includegraphics[width=6cm]{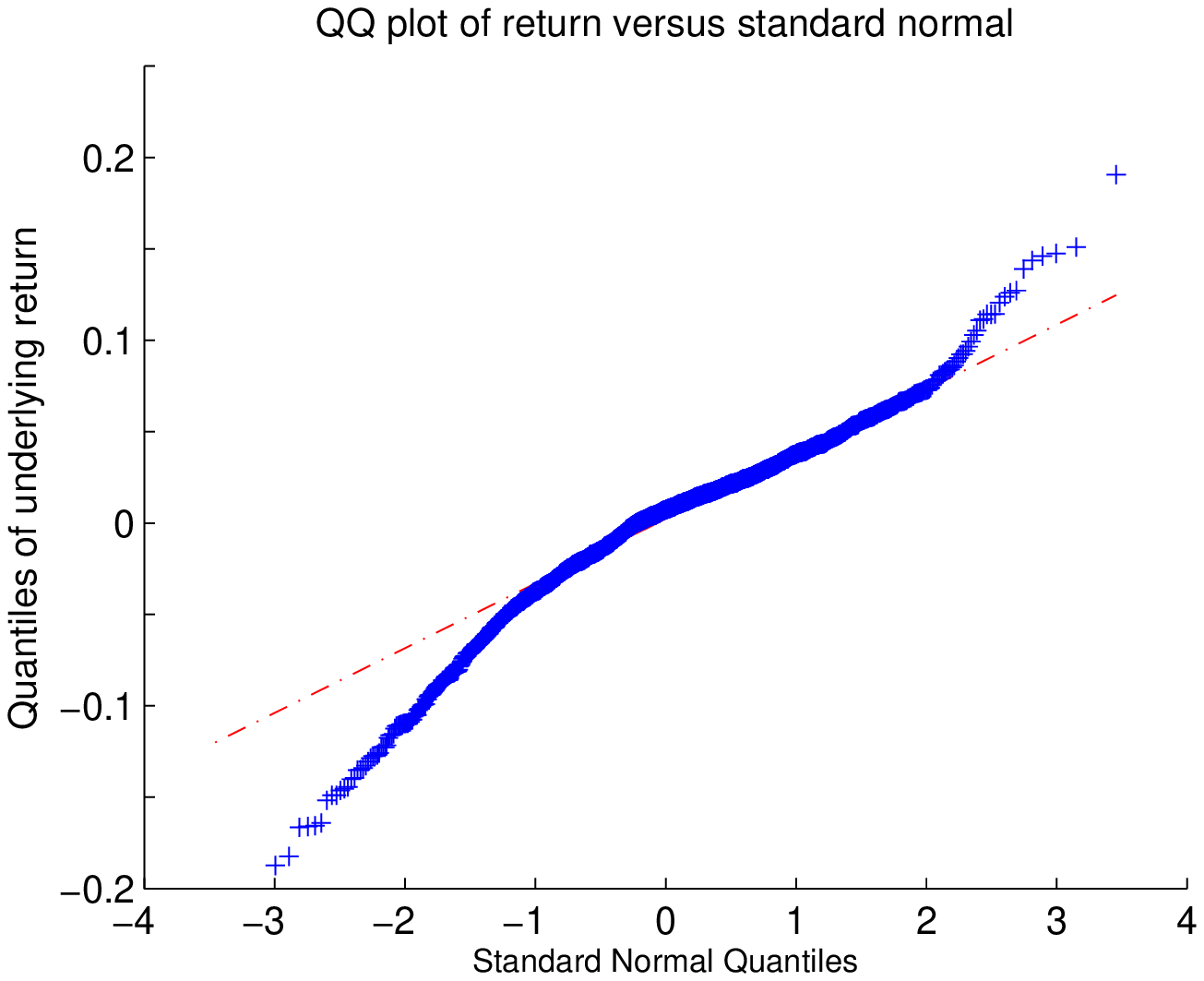}\quad
\includegraphics[width=6cm]{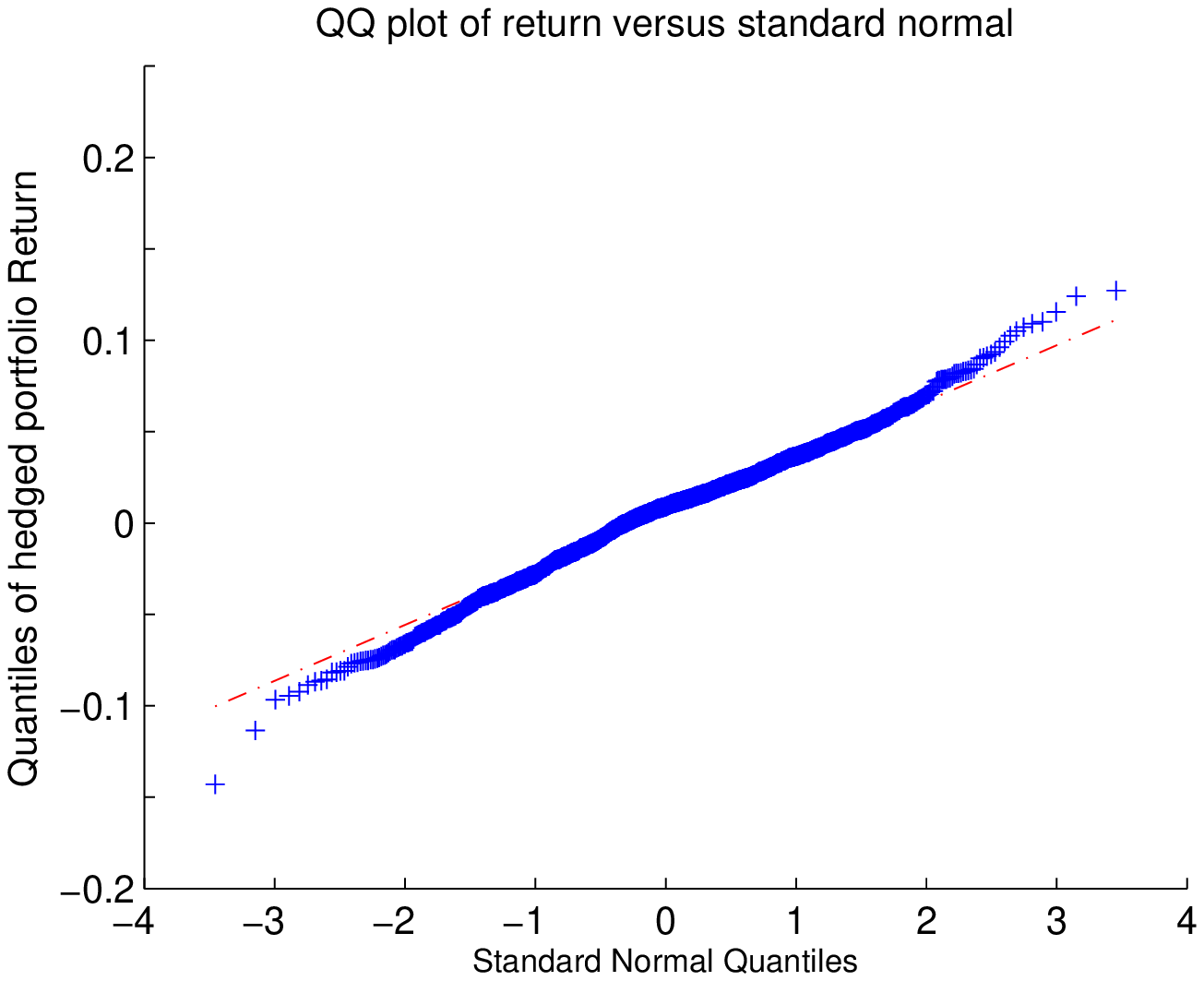}\\
\includegraphics[width=6cm]{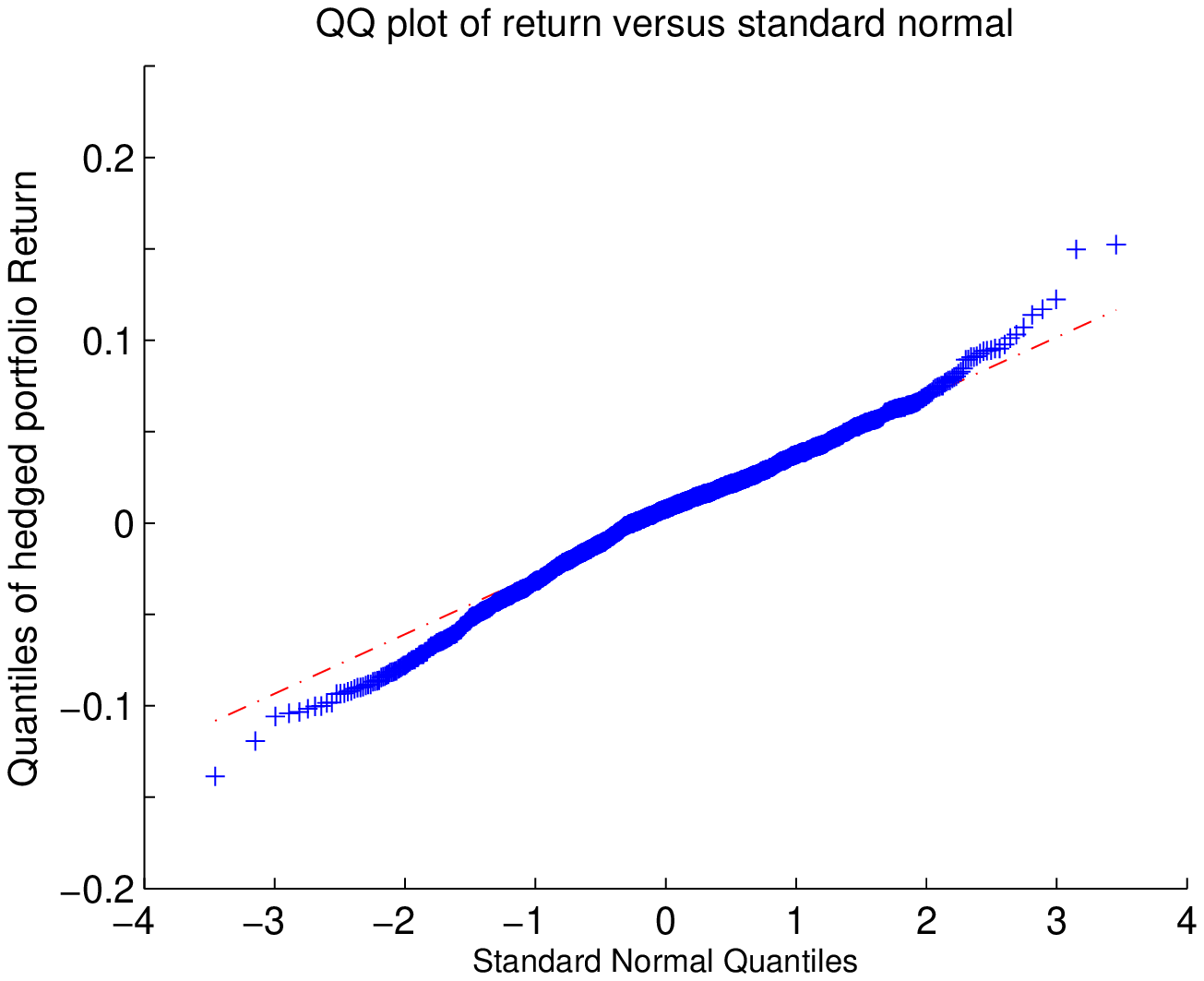}\quad
\includegraphics[width=6cm]{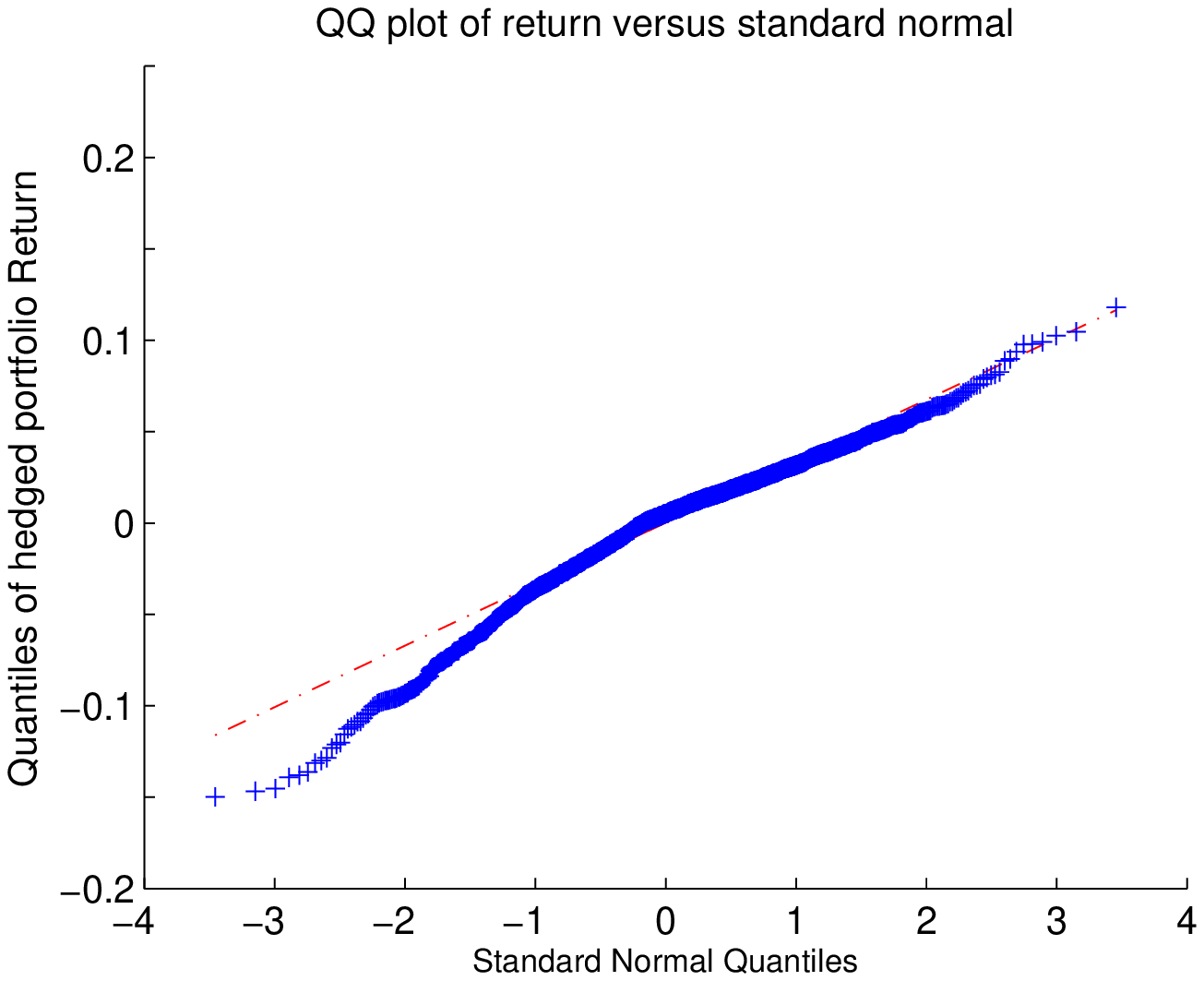}
\end{center}
\caption{QQ-plots for S\&P 500 return distribution versus normal (top-left), the return distribution of hedged portfolio with realized third moment variation (top-right), the return distribution of hedged portfolio with realized fourth moment variation (bottom-left) and the return distribution of hedged portfolio with realized Neuberger third moment (bottom-right) versus normal}
\label{Fig:QQ_SP}
\end{figure}

In the top-left panel, the skewed and heavy tailed distribution of S\&P 500 return is presented.
In the top-right panel, it is shown that the hedged portfolio has more Gaussian-like return distribution.
Similar properties are demonstrated in the bottom panels where the QQ-plots for the hedged portfolio with the fourth moment variation (left) and Neuberger third moment (right) are plotted.
In terms of RMSE, the most Gaussian-like return distribution is for the portfolio hedged with the third moment variation as presented in Table~\ref{Table:RMSE}.
The portfolios hedged by fourth moment and the Neuberger third moment also show modest performances.

\begin{table}
\caption{RMSE between the quantile of the return distribution of hedged portfolio and normal}\label{Table:RMSE}
\begin{center}
\begin{tabular}{cccc}
\midrule
& third moment & fourth moment & Neuberger \\
\midrule
RMSE & 0.0030 & 0.0042 & 0.0067\\
\midrule
\end{tabular}
\end{center}
\end{table}

The hedge price, the fixed leg of the swap, may be determined by two counterparties of the swap contract in the market as an interaction of supply and demand.
However, the risk-neutral expectations of the moments derived in Theorem~\ref{Thm:representation}, the prices of portfolios of European call and put options with specific weights, would be a theoretical value as in the case of variance swap.

\section{Conclusion}\label{sect:concl}
We define the third and fourth moment variations to deal with high moments properties of return distribution.
The realized moment variations are good approximations to the actual moments of return distribution.
Therefore, one can use the third or fourth moment variation swap to hedge tail risk.
The return of the hedged portfolio with the swap follows more Gaussian-like distribution and hence the investor with hedged position is more robust to tail risk.
We also derive the risk-neutral expectations of the moment variations in terms of European option prices.

\appendix
\section{Proof of Theorem~\ref{Thm:representation}}\label{Apped:proof}
To prove Theorem~\ref{Thm:representation}, we need some preliminary results.
The following lemma is an extended version of Proposition 1 in \cite{CarrWu} and the lemma for continuous processes is introduced by \cite{Lee2}.
\begin{lemma}\label{Lemma:replication2}
Let $X^c$ be continuous part of $X$.
If $g$ is a continuous function with its anti-derivative $G$ and second anti-derivative $\bar G$, then
\begin{align*}
&\int_{u}^{t} g(X_{s-}) \D [X^c]_s + 2\sum_{u \leq s \leq t} [\Delta \bar G(X_s)  - \Delta X_s G(X_{s-})] \\
={}& 2 \left( \int_{u}^{t} (G(X_u) - G(X_s)) \D X_s + \int_{X_u}^{X_t} g(K)(X_t - K) \D K \right)\\
={}& 2 \left( \int_{u}^{t} (G(X_u) - G(X_s)) \D X_s \right.
+ \left.  \int_{0}^{X_u} g(K) (K-X_t)^+ \D K + \int_{X_u}^{\infty} g(K) (X_t-K)^+ \D K \right).
\end{align*}
\end{lemma}

\begin{proof}
If $f$ is twice continuously differentiable, then by It\^{o}'s lemma for a semimartingale,
$$f(X_t) = f(X_u) + \int_{u}^{t} f'(X_{s-}) \D X_s + \frac{1}{2}\int_{u}^{t} f''(X_{s-}) \D [X^c]_s + \sum_{u < s \leq t}[\Delta f(X_s) - \Delta X_s f'(X_{s-})] $$
and by Taylor's theorem with the integral form of the remainder term
\begin{align*}
f(X_t) ={}& f(X_u) + f'(X_u)(X_t - X_u) + \int_{X_u}^{X_t} f''(K)(X_t - K) \D K \\
={}& f(X_u) + f'(X_u)(X_t - X_u) \\
+&\int_{0}^{X_u} f''(K)(K-X_t)^+ \D K + \int^{\infty}_{X_u} f''(K)(X_t-K)^+ \D K.
\end{align*}
By comparing above equations, we have
\begin{align*}
&\int_{u}^{t} f''(X_{s-}) \D [X^c]_s + 2\sum_{u \leq s \leq t}[\Delta f(X_s) - \Delta X_s f'(X_{s-})]\\
&=  2 \left( \int_{u}^{t} (f'(X_u) - f'(X_s)) \D X_s + \int_{X_u}^{X_t} f''(K)(X_t - K) \D K  \right) \\
&=  2 \left( \int_{u}^{t} (f'(X_u) - f'(X_s)) \D X_s  \right. +\left.  \int_{0}^{X_u} f''(K)(K-X_t)^+ \D K + \int^{\infty}_{X_u} f''(K)(X_t-K)^+ \D K \right).
\end{align*}
Finally, substituting $g$ in place of $f''$, we have the desired result.
\end{proof}

For a continuous process $X$, the above lemma is simply
\begin{align*}
&\int_{u}^{t} g(X_{s-}) \D [X]  \\
={}& 2 \left( \int_{u}^{t} (G(X_u) - G(X_s)) \D X_s \right.
+ \left.  \int_{0}^{X_u} g(K) (K-X_t)^+ \D K + \int_{X_u}^{\infty} g(K) (X_t-K)^+ \D K \right).
\end{align*}
If $X$ is a martingale, the stochastic integral at the right-hand-side vanishes by taking expectation under certain conditions.
The integrands in two Riemann integrals represents the payoff of European put and call options.
In the next, we represent the $\Q$-expectation of a stochastic integral with respect to $[F]$ as a combination of weighted call and put option prices.
There is also an interesting application to derivative pricing using the following lemma, see~\cite{Lee}.

\begin{lemma}\label{Lemmma:representation}
For a continuous function $g$ with its anti-derivative $G$ and second anti-derivative $\bar G$, if $G(F) \in L^2_{\Q, [F]}$, then we have
\begin{align*}
&\E^{\Q} \left[ \int_{0}^{T} g \left( F_{s-} \right) \D \left[F^c\right]_s  + 2\sum_{o < s \leq T} [\bar G(F_s) -\bar G(F_{s-}) - \Delta F_s G(F_{s-})] \right] \\
& \quad = 2\e^{r T} \int_{0}^{\infty} g(K) \phi \left( S_0, K \right) \D K.
\end{align*}
\end{lemma}

\begin{proof}
Take $\Q$-expectation to the result of Lemma~\ref{Lemma:replication2} with $X = F$.
\end{proof}

\begin{lemma}\label{Lemma:representation2}
Let $f(x)$ and $g(x)$ are twice continuously differentiable functions, $H$ be a anti-derivative of $f'g'$ and $\bar H$ is a second anti-derivative of $f'g'$.
If $H(F) \in L^2_{\Q, [F]}$, then
\begin{align*}
\E^{\Q}\left[\,[f(F), g(F)]_T \right] ={}& f(F_0)g(F_0) + 2e^{rT}  \int_{0}^{\infty} \frac{\D f}{\D x}(K)\frac{\D g}{\D x}(K) \phi(K) \D K\\
&+ \E^{\Q} \left[ \sum_{0<s\leq T}\{ \Delta f(F_s)\Delta g(F_s) - 2\Delta \bar H(F_s) + 2\Delta F_s H(F_{s-}) \} \right].
\end{align*}
\end{lemma}

\begin{proof}
By It\'{o}'s formula,
$$f(F_t) = f(F_u) + \int_{u}^{t} f'(F_{s-}) \D F^c_s + \frac{1}{2}\int_{u}^{t} f''(F_{s-}) \D [F^c]_s + \sum_{u < s \leq t}\Delta f(F_s), $$
$$g(F_t) = g(F_u) + \int_{u}^{t} g'(F_{s-}) \D F^c_s + \frac{1}{2}\int_{u}^{t} g''(F_{s-}) \D [F^c]_s + \sum_{u < s \leq t}\Delta g(X_s). $$
First,
$$\left[ \int_{u}^{t} f'(F_{s-}) \D F^c_s, \int_{u}^{t} g'(F_{s-}) \D F^c_s \right] = \int_{u}^{t} f'(F_{s-}) g'(F_{s-}) \D [F^c]_s .$$
Second,
$$\left[ \sum_{u < s \leq t}\Delta f(F_s), \sum_{u < s \leq t}\Delta g(X_s) \right] = \sum_{u < s \leq t} (\Delta f(F_s)\Delta g(X_s)).$$
Note that the quadratic covariation between pure jump process and continuous process is zero.
Therefore,
$$ [f(F), g(F)]_t =  f(F_0)g(F_0) + \int_{0}^{t} \frac{\D f}{\D x}(F_{s-})\frac{\D g}{\D x}(F_{s-}) \D [F^c]_s + \sum_{0<s\leq t}(\Delta f(F_s) \Delta g(F_s)).$$
Applying Lemma~\ref{Lemmma:representation}, we have
\begin{align*}
\E [f(F), g(F)]_t ={}& f(F_0)g(F_0) + \E \left[ \int_{0}^{t} \frac{\D f}{\D x}(F_{s-})\frac{\D g}{\D x}(F_{s-}) \D [F^c]_s  + \sum_{0<s\leq t}(\Delta f(F_s) \Delta g(F_s))\right] \\
={}& f(F_0)g(F_0) + 2e^{rT}  \int_{0}^{\infty} \frac{\D f}{\D x}(K)\frac{\D g}{\D x}(K) \phi(K) \D K \\
&+ \E \left[ \sum_{0<s\leq t}( - 2 \Delta \bar H(F_s)  + 2\Delta F_s H(F_{s-}) +  \Delta f(F_s) \Delta g(F_s))\right],
\end{align*}
and this complete the proof.
\end{proof}

\begin{remark}
The above result also holds when $f$ and $g$ are functions of $t$, since the covariation between $t$, finite variation and continuous, and other stochastic process is zero.
In other words, the lemma is also applicable for $[f(t,F), g(t,F)]_t$.
Since we have the same result, for the notational simplicity, in the next, we use $f(F)$ instead of $f(t, F)$.
\end{remark}

Now we start to prove the Theorem~\ref{Thm:representation}.
For $\E^{\Q}\left[\, [R]_T \right]$, it is well known (see \cite{CarrWu}).
For the option parts for $[R, R^2]$ and $[R^2]$, use the fact that $R_t = \log F_t/F_0 - rt$ and apply Lemma~\ref{Lemma:representation2}.
In the case of $[R, R^2]_T$, $f(x) = \log(x/F_0) + rt$ and $g(x) = \log^2(x/F_0) + 2rt\log(x/F_0) + r^2 t^2$.
In addition,
$ f'(x) = 1/x $ and $ g'(x) = 2/K \log (x/S_0)$.
(More precisely, $g$ is a function of $t$ and $x$ and we use $g'(T,x) = 2/K \log (x/S_0)$.)
Thus,
$$ \frac{\D f}{\D x}(K)\frac{\D g}{\D x}(K) = \frac{2}{K^2} \log \frac{K}{S_0} .$$
Similarly for $[R^2]_T$, $f(x) = g(x) = \log^2(x/F_0) + 2rt\log(x/F_0) + r^2 t^2$.
In this case, we have
$$ \frac{\D f}{\D x}(K)\frac{\D g}{\D x}(K) = \frac{4}{K^2} \left(\log \frac{K}{S_0} \right)^2.$$

For the jump correction term $J_3$ of $[R, R^2]$, we have
\begin{align*}
\E^{\Q} \left[ \sum_{0<s\leq T} \Delta f(F_s) \Delta g(F_s) \right] &= \E^{\Q} \left[ \sum_{0<s\leq T} \Delta R_s \Delta \left( R^2_s \right) \right]= \E^{\Q} \left[ \sum_{0<s\leq T} \Delta R_s  \left( R^2_s - R^2_{s-} \right) \right]\\
&= \E^{\Q} \left[ \sum_{0<s\leq T} \Delta R_s  \left( R_s - R_{s-} \right)(R_s + R_{s-}) \right]\\
&= \E^{\Q} \left[ \sum_{0<s\leq T} (\Delta R_s)^2 (\Delta R_s + R_{s-} + R_{s-}) \right]\\
&= \E^{\Q}\left[ \int_{[0,T]\times \mathbb R}\left( 2x^2 R_{s-} + x^3 \right)J(\D s \times \D x) \right]
\end{align*}
where $J$ is the jump measure for the log-return process $R$.
In addition, the antiderivative and second antiderivative of $f'g'$ are
$$H(F) = -2\frac{\log(F/F_0) + rt + 1}{F}, \quad \bar H(F) = -\log^2(F/F_0) - 2(1+rt)\log F.$$
Note that
\begin{align*}
H(F_{s-})\Delta F_s &= -2 \frac{R_{s-} + 1}{F_{s-}}\Delta F_s = -2(R_{s-} + 1)\left(\frac{F_s}{F_{s-}}-1\right)\\
&=-2(R_{s-}+1)(\e^{\Delta R_s}-1)
\end{align*}
and
\begin{align*}
\Delta \bar H(F_s) &= - \left\{ \log^2 \frac{F_s}{F_0} - \log^2 \frac{F_{s-}}{F_0} + 2(1+rs) \left(\log \frac{F_s}{F_0} - \log\frac{F_{s-}}{F_0}\right) \right\} \\
&= - \left\{ \left(\log\frac{F_{s}}{F_0} - \log\frac{F_{s-}}{F_0} \right) \left( 2\log\frac{F_{s-}}{F_0} + \log\frac{F_s}{F_0} - \log\frac{F_{s-}}{F_0} \right) + 2(1+rs) \left(\log \frac{F_s}{F_{s-}} \right) \right\} \\
&= -\Delta R_s \left( 2(R_{s-} -rs) + \Delta R_s\right) - 2(1+rs) \Delta R_s \\
&= -\Delta R_s (2R_{s-} + \Delta R_s + 2 ).
\end{align*}
Thus,
$$\E^{\Q}\left[ \sum_{0<s\leq T} 2H(F_{s-})\Delta F_s \right] = - \E^{\Q} \left[ \int_{[0,T]\times \mathbb R} 4(\e^x-1)\left(R_{s-} + 1 \right)J(\D s \times \D x) \right]$$
and
$$\E^{\Q}\left[ \sum_{0<s\leq T} 2\Delta \bar H(F_{s}) \right] = - \E^{\Q} \left[ \int_{[0,T]\times \mathbb R} 2\left( 2x R_{s-} + x^2 + 2x \right)J(\D s \times \D x) \right].$$
Therefore,
\begin{align*}
&\E^{\Q} \left[ \sum_{0<s\leq T}\{\Delta f(F_s)\Delta g(F_s) - 2\Delta \bar H(F_s) + 2H(F_{s-})\Delta F_s \} \right] \\
&= 4 \E^{\Q}  \left[ \int_{[0,T]\times \mathbb R} \left( R_{s-} \left( 1 + x + \frac{1}{2}x^2 - \e^x \right) + 1 + x + \frac{1}{2}x^2+\frac{1}{4}x^3 - \e^x  \right) J(\D s \times \D x) \right].
\end{align*}

For the jump correction term $J_4$ of $[R^2]$, let
$ f(F) = g(F) = \log^2(F/F_0) + 2rt\log(F/F_0) + r^2t^2$.
Then
$$ f'(F)g'(F) = \frac{4}{F^2} \left( \log\frac{F}{F_0} + rt \right)^2$$
and
$$ H(F) = -4\frac{ \log^2(F/F_0) + 2(rt+1)\log(F/F_0) + r^2 t^2 + 2rt + 2}{F} $$
and
$$ \bar H(F) = -\frac{4}{3}\log^2\frac{F}{F_0}\left(\log\frac{F}{B_0} + 3rt +3 \right) - 4\log(F)(r^2t^2+2rt+2).$$
In addition,
\begin{align*}
\Delta f(F_s) \Delta g(F_s) &= (R^2_s - R^2_{s-})^2 = (R_s - R_{s-})^2(R_s+R_{s-})^2 = (\Delta R_s)^2 (2R_{s-} + \Delta R_s)^2 \\
&= (\Delta R_s)^2(4R^2_{s-} + 4\Delta R_s R_{s-} + (\Delta R_s)^2 )
\end{align*}
and
\begin{align*}
H(F_{s-})\Delta F_s &= -4\left\{ \left( \log\frac{F_{s-}}{F_0} + rt \right)^2 + 2\left(\log\frac{F_{s-}}{F_0} + rt \right) + 2\right\}\left( \frac{F_s}{F_{s-}} - 1 \right) \\
&= -4(R^2_{s-} + 2R_{s-} + 2)(\e^{\Delta R} -1).
\end{align*}

To compute $\Delta \bar H(F_s)$, we need
\begin{align*}
\log^3 \frac{F_s}{F_0} - \log^3 \frac{F_s}{F_0} &= \left(\log \frac{F_{s}}{F_{0}} - \log \frac{F_{s-}}{F_{0}}\right)\left(\log^2\frac{F_s}{F_0} + \log^2\frac{F_{s-}}{F_0} + \log\frac{F_s}{F_0}\log\frac{F_{s-}}{F_0}\right)\\
&= \Delta R_s( (\Delta R_s + R_{s-}-rs)^2 + (R_{s-}-rs)^2 + (\Delta R_s + R_{s-}-rs)(R_{s-}-rs) )\\
&= \Delta R_s( (\Delta R_s)^2 + 3R^2_{s-} + 3r^2s^2 + 3\Delta R_s R_{s-} - 3rs\Delta R_s - 6rsR_{s-})
\end{align*}
and
$$
(rs+1)\left( \log^2\frac{F_s}{F_0} - \log^2\frac{F_{s-}}{F_0} \right) = \Delta R_s(2rsR_{s-} - 2r^2s^2 + rs\Delta R_s + 2R_{s-} - 2rs + \Delta R_s)
$$
and
$$(r^2s^2 + 2rs+2)(\log F_s - \log F_{s-}) = \Delta R_s (r^2s^2 + 2rs +2).$$
Using above equations, we have
$$ \Delta \bar H(F_s) = -4\Delta R_s \left( \frac{(\Delta R_s)^2}{3} + R_{s-}\Delta R_s + \Delta R_s + R_{s-}^2 + 2R_{s-} +2 \right).$$
Finally,
\begin{align*}
&\E^{\Q} \left[ \sum_{0<s\leq T} \{ \Delta f(F_s)\Delta g(F_s) - 2\Delta \bar H(F_s) + 2H(F_{s-})\Delta F_s \} \right]\\
={}& 8 \E^{\Q} \left[ \int_{[0,T]\times \mathbb R} \left( R^2_{s-}\left( 1 + x + \frac{1}{2}x^2 - \e^x \right) + 2R_{s-}\left(1 + x + \frac{1}{2}x^2 + \frac{1}{4}x^3 - \e^x \right) \right. \right. \\
&\left. + 2\left(  1 + x + \frac{1}{2}x^2 +  \frac{1}{6}x^3 + \frac{1}{16}x^4  -\e^x \right) J(\D s \times \D x) \right].
\end{align*} 


\begin{thebibliography}{}

\bibitem[Andersen et~al., 2003]{ABDL}
Andersen, T.~G., Bollerslev, T., Diebold, F.~X., \& Labys, P. (2003).
\newblock Modeling and forecasting realized volatility.
\newblock {\em Econometrica}, 71, 579--625.

\bibitem[\'{A}ngel Le\'{o}n et~al., 2005]{LeonRubioSerna}
\'{A}ngel Le\'{o}n, Rubio, G., \& Serna, G. (2005).
\newblock Autoregresive conditional volatility, skewness and kurtosis.
\newblock {\em The Quarterly Review of Economics and Finance}, 45, 599 -- 618.

\bibitem[Bakshi et~al., 2003]{BakshiKapadiaMadan}
Bakshi, G., Kapadia, N., \& Madan, D. (2003).
\newblock Stock return characteristics, skew laws, and the differential pricing
  of individual equity options.
\newblock {\em Review of Financial Studies}, 16, 101--143.

\bibitem[Bakshi and Madan, 2006]{BakshiMadan}
Bakshi, G. \& Madan, D. (2006).
\newblock A theory of volatility spreads.
\newblock {\em Management Science}, 52, 1945--1956.

\bibitem[Barndorff-Nielsen and Shephard, 2002]{Barndorff2002a}
Barndorff-Nielsen, O.~E. \& Shephard, N. (2002).
\newblock Econometric analysis of realized volatility and its use in estimating
  stochastic volatility models.
\newblock {\em Journal of the Royal Statistical Society: Series B (Statistical
  Methodology)}, 64, 253--280.

\bibitem[Barndorff-Nielsen and Shephard, 2004]{Barndorff-Nielsen2004}
Barndorff-Nielsen, O.~E. \& Shephard, N. (2004).
\newblock Power and bipower variation with stochastic volatility and jumps.
\newblock {\em Journal of Financial Econometrics}, 2, 1--37.

\bibitem[Barndorff-Nielsen and Shephard, 2006]{Barndorff-Nielsen2006}
Barndorff-Nielsen, O.~E. \& Shephard, N. (2006).
\newblock Econometrics of testing for jumps in financial economics using
  bipower variation.
\newblock {\em Journal of Financial Econometrics}, 4, 1--30.

\bibitem[Britten-Jones and Neuberger, 2000]{Britten-Jones}
Britten-Jones, M. \& Neuberger, A. (2000).
\newblock Option prices, implied price processes, and stochastic volatility.
\newblock {\em The Journal of Finance}, 55, 839--866.

\bibitem[Brooks et~al., 2005]{Brooks}
Brooks, C., Burke, S.~P., Heravi, S., \& Persand, G. (2005).
\newblock Autoregressive conditional kurtosis.
\newblock {\em Journal of Financial Econometrics}, 3, 399--421.

\bibitem[Carr and Madan, 2001]{CarrMadan}
Carr, P. \& Madan, D. (2001).
\newblock Towards a theory of volatility trading.
\newblock In {\em Handbooks in Mathematical Finance: Option Pricing, Interest
  Rates and Risk Management}, 458--476. Cambridge University Press.

\bibitem[Carr and Wu, 2009]{CarrWu}
Carr, P. \& Wu, L. (2009).
\newblock Variance risk premiums.
\newblock {\em Review of Financial Studies}, 22, 1311--1341.

\bibitem[Christoffersen et~al., 2006]{Christoffersen}
Christoffersen, P., Heston, S., \& Jacobs, K. (2006).
\newblock Option valuation with conditional skewness.
\newblock {\em Journal of Econometrics}, 131, 253 -- 284.

\bibitem[Demeterfi et~al., 1999]{DDKZ}
Demeterfi, K., Derman, E., Kamal, M., \& Zou, J. (1999).
\newblock More than you ever wanted to know about volatility swaps.
\newblock Technical report, Goldman Sachs.

\bibitem[Du and Kapadia, 2012]{DuKapadia}
Du, J. \& Kapadia, N. (2012).
\newblock Volatility and tail indices from option prices.
\newblock {\em Working paper, University of Massachusetts, Amherst}.

\bibitem[Hansen and Lunde, 2006]{Hansen}
Hansen, P.~R. \& Lunde, A. (2006).
\newblock Realized variance and market microstructure noise.
\newblock {\em Journal of Business and Economic Statistics}, 24, 127--161.

\bibitem[Harvey and Siddique, 2000]{HarveySiddique}
Harvey, C. \& Siddique, A. (2000).
\newblock Conditional skewness in asset pricing tests.
\newblock {\em Journal of Finance}, 55, 1263--1295.

\bibitem[Harvey and Siddique, 1999]{HarveySiddique1999}
Harvey, C.~R. \& Siddique, A. (1999).
\newblock Autoregressive conditional skewness.
\newblock {\em The Journal of Financial and Quantitative Analysis},
  34, 465--487.

\bibitem[Kozhan et~al., 2012]{Kozhan}
Kozhan, R., Neuberger, A., \& Schneider, P. (2012).
\newblock The skew risk premium in the equity index market.
\newblock {\em Working paper, Warwick Business School}.

\bibitem[Kraus and Litzenberger, 1976]{KrausLitzenberger}
Kraus, A. \& Litzenberger, R.~H. (1976).
\newblock Skewness preference and the valuation of risk assets.
\newblock {\em The Journal of Finance}, 31, 1085--1100.

\bibitem[Kuo, 2006]{Kuo}
Kuo, H.-H. (2006).
\newblock {\em Introduction to Stochastic Integration}.
\newblock New York: Springer.

\bibitem[Lee, 2012a]{Lee2}
Lee, K. (2012a).
\newblock {Unpublished doctoral dissertation, \em GARCH intensity model and new methods of option pricing}.
\newblock KAIST, Daejon.

\bibitem[Lee, 2012b]{Lee}
Lee, K. (2012b).
\newblock Recursive formula for arithmetic asian option prices.
\newblock {\em Journal of Futures Markets}, 
  http://dx.doi.org/10.1002/fut.21591.

\bibitem[Mykland and Zhang, 2009]{MyklandZhang}
Mykland, P.~A. \& Zhang, L. (2009).
\newblock Inference for continuous semimartingales observed at high frequency.
\newblock {\em Econometrica}, 77, 1403--1445.

\bibitem[Neuberger, 2012]{Neuberger2012}
Neuberger, A. (2012).
\newblock Realized skewness.
\newblock {\em Review of Financial Studies}, 25, 3423--3455.

\bibitem[Protter, 2005]{Protter}
Protter, P.~E. (2005).
\newblock {\em Stochastic integration and differential equations}.
\newblock Springer.

\bibitem[Todorov, 2010]{Todorov}
Todorov, V. (2010).
\newblock Variance risk-premium dynamics: The role of jumps.
\newblock {\em Review of Financial Studies}, 23, 345--383.

\bibitem[Zhang, 2012]{Zhang}
Zhang, L. (2012).
\newblock Implied and realized volatility: empirical model selection.
\newblock {\em Annals of Finance}, 8, 259--275.

\end{thebibliography}
\end{document}